\documentclass[11pt]{article} 

\usepackage[utf8]{inputenc}
\usepackage[pdftex,left=1in,top=1in,bottom=1in,right=1in]{geometry}
\usepackage{amsmath, amssymb, dsfont, mathrsfs,amsthm}
\usepackage{braket}
\usepackage{graphicx}
\usepackage{caption}
\usepackage{bm}
\usepackage{xcolor}
\usepackage[framemethod=tikz]{mdframed} 
\usepackage[colorlinks=true,citecolor=blue,linkcolor=blue]{hyperref}
\usepackage[nameinlink, noabbrev, capitalize]{cleveref}
\usepackage{physics}
\usepackage{verbatim} 
\usepackage{autobreak}
\usepackage{enumitem}
\usepackage{mathtools}
\hypersetup{
    colorlinks,
    linkcolor={red!50!black},
    citecolor={blue!50!black},
    urlcolor={blue!80!black}
}

\newtheorem{theorem}{Theorem}
 
 \newtheorem{claim}{Claim}
 \newtheorem{lemma}{Lemma}
 
 \newtheorem{corollary}{Corollary}
 
 \newtheorem{definition}{Definition}

\newcommand{\C}{\mathds{C}}
\newcommand{\N}{\mathds{N}}

\newcommand{\F}{\mathds{F}}

\newcommand{\R}{\mathbb{R}}

\newcommand{\eps}{\epsilon}
\newcommand{\samp}{\leftarrow} 
\newcommand{\trd}[2]{\norm{#1 - #2}_{Tr}} 
\newcommand{\supp}{\mathsf{supp}}

\newcommand{\zo}{\{0, 1\}}
\newcommand{\regi}{\mathsf{R}} 
\newcommand{\regis}[1]{\mathsf{R}_{\mathsf{#1}}} 
\newcommand{\adve}{\mathcal{A}}
\newcommand{\negl}{\mathsf{negl}}
\newcommand{\poly}{\mathsf{poly}}
\newcommand{\hyb}{\mathsf{Hyb}}

\newcommand{\io}{i\mathcal{O}}
\newcommand{\Setup}{\mathsf{Setup}}
\newcommand{\Ver}{\mathsf{Ver}}

\newcommand{\fsplit}{\mathsf{MallPunc}}
\newcommand{\extprob}{2^{-k}} 
\newcommand{\outtrue}{\mathsf{TRUE}}
\newcommand{\outfalse}{\mathsf{FALSE}}
\newcommand{\prf}{\mathsf{PRF}}
\newcommand{\keygen}{\mathsf{KeyGen}}
\newcommand{\ceval}{\mathsf{Eval}}
\newcommand{\obfd}[1]{\hat{#1}}
\newcommand{\indtag}{0^{\ell_4(\lambda)}} 
\newcommand{\enc}{\mathsf{Enc}}
\newcommand{\dec}{\mathsf{Dec}}
\newcommand{\fail}{\bot}
\newcommand{\ext}{\mathsf{Ext}}
\newcommand{\game}{\mathsf{Game}}
\newcommand{\gamewp}[1]{\game_{#1}(1^\lambda)}
\newcommand{\codeof}[1]{\langle{#1}\rangle}

\newcommand{\cmoe}{c_{MoE}}
\newcommand{\copyprot}{\mathsf{QuantumProtect}}
\newcommand{\genstate}{\mathsf{GenState}}
\newcommand{\prot}{\mathsf{Protect}}
\newcommand{\qpeval}{\mathsf{Eval}}

\newcommand{\canonvec}[2]{\mathsf{Can}(#1, #2)}

\newcommand{\chal}{\mathsf{Chal}}
\newcommand{\sampchal}{\mathsf{SampCh}}
\newcommand{\sampinp}{\mathcal{D}_{\mathsf{inp}}}
\newcommand{\sampchalfrominp}{\mathsf{SampChFromInp}}
\newcommand{\punc}{\mathsf{Punc}}
\newcommand{\ptriv}{p^{\mathcal{F}}_{triv}}
\newcommand{\ti}{\mathsf{TI}}
\newcommand{\tip}[2]{\ti_{#2, #1}}
\newcommand{\ati}{\mathsf{ATI}}
\newcommand{\simati}{\mathsf{SimATI}}
\newcommand{\simatiwp}[6]{\simati^{#2, #1, #6}_{#3, #5}}
\newcommand{\simatiwpnop}[5]{\simati^{#2, #1, #5}_{#4}}
\newcommand{\atip}[4]{\ati^{#4, #2}_{#3, #1}}

\newcommand{\funckey}{k}
\newcommand{\cscheme}{\mathcal{CS}}
\newcommand{\csplit}{C_{\mathsf{punc}}}
\newcommand{\cpred}{Q_{\mathsf{rel}}}

\newcommand{\ace}{\mathsf{ACE}}
\newcommand{\stegenc}{\mathsf{StegEnc}}
\newcommand{\stegdec}{\mathsf{StegDec}}
\newcommand{\acepunchidinggame}{\mathsf{PunctureHidingGame}_{\adve}(1^\lambda)}
\newcommand{\aceprct}{\mathsf{PRCiphertextGame}_{\adve}(1^\lambda)}
\newcommand{\acesteg}{\mathsf{StegCiphertextGame}_{\adve}(1^\lambda)}
\newcommand{\acectcount}{{\ell(\lambda)}}
\newcommand{\acemeslen}{{n(\lambda)}}
\newcommand{\acecirclen}{{c_{punc}(\lambda)}}
\newcommand{\acedistsize}{{c_{samp}(\lambda)}}
\newcommand{\acestegbound}{{t_{steg}(\lambda)}}
\newcommand{\genek}{{\mathsf{GenEK}}}
\newcommand{\gendk}{{\mathsf{GenDK}}}
\newcommand{\aceselen}{{\mathsf{se(\lambda)}}}
\newcommand{\acelimit}{{t_{limit}}}
\newcommand{\aceminent}{{k(\lambda)}}

\newcommand{\cprotectgame}{\mathsf{AntiPiracyGame}}
\newcommand{\cprotectgamewp}{\cprotectgame_{\copyprot,\adve}(1^\lambda)}

\newcommand{\strongcprotectgame}{\mathsf{StrongAntiPiracyGame}}
\newcommand{\strongcprotectgamewp}{\strongcprotectgame^{\gamma(\lambda)}_{\copyprot,\adve}(1^\lambda)}
\newcommand{\secgame}{\mathsf{SecurityGame}}
\newcommand{\secgamewp}[1]{\secgame_{\adve}(1^\lambda)}

\newcommand{\puncgame}{\mathsf{MalleablePuncturingGame}}
\newcommand{\puncgamewp}[1]{\puncgame_{#1}(1^\lambda)}

\title{How to Copy-Protect Malleable-Puncturable Cryptographic Functionalities Under Arbitrary Challenge Distributions}

\date{}

 \author{Alper \c{C}akan\footnote{Part of this work was done while the author was an intern at NTT Research.}\\Carnegie Mellon University \\ \texttt{alpercakan98@gmail.com}  \and Vipul Goyal\\NTT Research \& Carnegie Mellon University\\ \texttt{vipul@vipulgoyal.org }}
\begin{document}	
\maketitle
\begin{abstract}
A quantum copy-protection scheme (Aaronson, CCC'09) encodes a functionality into a quantum state such that given this state, no efficient adversary can create two (possibly entangled) quantum states that are both capable of running the functionality. There has been a recent line of works on constructing provably-secure copy-protection schemes for general classes of schemes in the plain model, and most recently the recent work of \c{C}akan and Goyal (IACR Eprint, 2025) showed how to copy-protect all cryptographically puncturable schemes with pseudorandom puncturing points.

In this work, we show how to copy-protect even a larger class of schemes. We define a class of cryptographic schemes called \emph{malleable-puncturable} schemes where the only requirement is that one can create a circuit that is capable of answering inputs at points that are \emph{unrelated} to the challenge in the security game but does not help the adversary answer inputs \emph{related} to the challenge. This is a flexible generalization of puncturable schemes, and can capture a wide range of primitives that was not known how to copy-protect prior to our work.

Going further, we show that our scheme is secure against arbitrary high min-entropy challenge distributions whereas previous work has only considered schemes that are punctured at pseudorandom points.
\end{abstract}

\newpage
\tableofcontents
\newpage

\section{Introduction}
Starting with the seminal work of Wiesner \cite{Wie83}, a long line of research has shown that quantum information can be extremely useful in cryptography.  One of the main areas where quantum information helps cryptography is achieving a task that is classically impossible task no matter what cryptographic assumptions are made. This line of applications of quantum information is also called \emph{quantum protection}. This notion can further be divided into four main categories: \emph{copy-protection} (\cite{Aar09}), \emph{secure leasing} (\cite{EC:AnaLaP21}), \emph{unbounded/LOCC leakage-resilience} (\cite{TCC:CGLR24}) and \emph{intrusion-detection} (\cite{TCC:CGLR24}).

A recent line of works (\cite{STOC:ColGun24, C:AnaBeh24}) have focused on constructing quantum copy-protection schemes for general classes of functionalities,  culminating in the recent work of \c{C}akan and Goyal (\cite{ccakan2025copy}) who showed how to copy-protect in the plain model all cryptographically puncturable schemes with pseudorandom puncturing points. In this work, we further advance the class of functionalities that can be copy-protected: We define a class of cryptographic schemes called \emph{malleable-puncturable} schemes where the only requirement is that one can create a circuit that is capable of answering inputs at points that are \emph{unrelated} to the challenge in the security game but does not help the adversary answer inputs \emph{related} to the challenge. This is an extremely flexible generalization of puncturable schemes, and can capture a wide range of primitives that was not known how to copy-protect prior to our work.

\begin{theorem}[\cref{thm:main}, informal]
For any functionality $\mathsf{F}$ that satisfies \emph{malleable-puncturing} security (\cref{def:punc}), there exists a quantum copy-protection scheme for $\mathsf{F}$ secure against arbitrary high min-entropy challenge distributions; assuming subexponentially secure iO and one-way functions.
\end{theorem}

\section{Preliminaries}
\subsection{Notation}
When $S$ is a set, we write $x \samp S$ to mean that $x$ is sampled uniformly at random from $S$. When $\mathcal{D}$ is a distribution, we write $x \samp \mathcal{D}$ to mean that $x$ is sampled from $\mathcal{D}$. Finally, we write $x \samp \mathcal{B}(\regi)$ or $\mathcal{B}(a)$ to mean that $x$ is set to a sample as output by the quantum or classical randomized algorithm $\mathcal{B}$ run on $\regi$ or $a$. We use similar notation for quantum registers, e.g., $\regi \samp \mathcal{B}(a)$.

We write \emph{QPT} to mean a stateful (unless otherwise specified) \emph{quantum polynomial time}. We write $\regi$ to mean a quantum register, which keeps a quantum state that will evolve when the register is acted on, and it can be entangled with other registers. We write $\regi \samp \rho$ to mean that the register $\regi$ is initialized with a sample from the quantum distribution (i.e. mixed state) $\rho$. We refer the reader to \cite{Nielsen_Chuang_2010} for a preliminary on quantum information and computation.

Unless otherwise specified, all of our cryptographical assumptions are implicitly post-quantum, e.g., \emph{one-way functions} means \emph{post-quantum secure one-way functions}.

\subsection{Cryptography}
\subsubsection{Puncturable Pseudorandom Functions}\label{sec:prf}
In this section, we recall puncturable pseudorandom functions.
\begin{definition}[\cite{STOC:SahWat14}]
    A puncturable pseudorandom function (PRF) is a family of functions $\{F: \zo^{c(
\lambda)} \times \zo^{m(\lambda)} \to \zo^{n(\lambda)}\}_{\lambda \in \N^+}$ with the following associated efficient algorithms.
    \begin{itemize}
        \item $\prf.\mathsf{Setup}(1^\lambda):$ Takes in the unary representation of the security parameter and outputs a key $K$ in $\zo^{c(\lambda)}$.
        \item $\prf.\ceval(K, x):$ Takes in a key $K$ and an input $x$, outputs an evaluation of the PRF.
        \item $\prf.\mathsf{Puncture}(K, S):$ Takes as input a key and a set $S \subseteq \zo^{m(\lambda)}$, outputs a punctured key.
    \end{itemize}
    
    We require the following.
    \paragraph{Correctness.} For all efficient distributions $\mathcal{D}(1^\lambda)$ over the power set $2^{\zo^{m(\lambda)}}$, we require
    \begin{equation*}
        \Pr[\forall x \not\in S~~\prf.\ceval(K_S, x) = F(K, x): \begin{array}{c}
              S \samp \mathcal{D}(1^\lambda) \\
              K \samp \prf.\Setup(1^\lambda) \\
              K_S \samp \prf.\mathsf{Puncture}(K, S)
        \end{array}] = 1.
    \end{equation*}
    \paragraph{Puncturing Security} We require that any stateful QPT adversary $\adve$ wins the following game with probability at most $1/2 + \negl(\lambda)$.
    \begin{enumerate}
        \item $\adve$ outputs a set $S$.
        \item The challenger samples $K \samp \prf.\Setup(1^\lambda)$ and $ K_S \samp \prf.\mathsf{Puncture}(K, S)$
        \item The challenger samples $b \samp \zo$. If $b = 0$, the challenger submits $K_S, \{F(K, x)\}_{x \in S}$ to the adversary. Otherwise, it submits $K_S, \{y_s\}_{s \in S}$ to  the adversary where $y_s \samp \zo^{n(\lambda)}$ for all $s \in S$.
        \item The adversary outputs a guess $b'$ and we say that the adversary has won if $b' = b$.
    \end{enumerate}
\end{definition}

\begin{theorem}[\cite{STOC:SahWat14, FOCS:Zhandry12}]\label{thm:puncprfexists}
Let $n(\cdot), m(\cdot), e(\lambda), k(\lambda)$ be efficiently computable functions.
\begin{itemize}
    \item If (post-quantum) one-way functions exist, then there exists a (post-quantum) puncturable PRF with input space $\zo^{m(\lambda)}$ and output space $\zo^{n(\lambda)}$.

\item If we assume subexponentially-secure (post-quantum) one-way functions exist, then for any $c > 0$, there exists a (post-quantum) $2^{-\lambda^c}$-secure\footnote{While the original results are for negligible security against polynomial time adversaries, it is easy to see that they carry over to subexponential security. Further, by scaling the security parameter by a polynomial and simple input/output conversions, subexponentially secure (for any exponent $c'$) one-way functions is sufficient to construct for any $c$ a puncturable PRF that is $2^{-\lambda^c}$-secure.} puncturable PRF against subexponential time adversaries with input space $\zo^{m(\lambda)}$ and output space $\zo^{n(\lambda)}$.
\end{itemize}
\end{theorem}

\subsubsection{Indistinguishability Obfuscation}
In this section, we introduce indistinguishability obfuscation.
\begin{definition}
    An indistinguishability obfuscation scheme $\io$ for a class of circuits $\mathcal{C} = \{\mathcal{C}_\lambda\}_\lambda$ satisfies the following.
    \paragraph{Correctness.} For all $\lambda, C \in \mathcal{C}_\lambda$ and inputs $x$,
    $\Pr[\Tilde{C}(x) = C(x): \Tilde{C} \samp \io(1^\lambda, C)] = 1$.

\end{definition}
    \paragraph{Security.} Let $\mathcal{B}$ be any QPT algorithm that outputs two circuits $C_0, C_1 \in \mathcal{C}$ of the same size, along with auxiliary information, such that $\Pr[\forall x ~ C_0(x)=C_1(x) : (C_0, C_1, \regis{aux}) \samp \mathcal{B}(1^\lambda)] \geq 1 - \negl(\lambda)$. Then, for any QPT adversary $\mathcal{A}$,
    \begin{align*}
      \bigg|&\Pr[\adve(\io(1^\lambda, C_0), \regis{aux}) = 1 :  (C_0, C_1, \regis{aux}) \samp \mathcal{B}(1^\lambda)] -\\ &\Pr[\adve(\io(1^\lambda, C_1), \regis{aux}) = 1 : (C_0, C_1, \regis{aux}) \samp \mathcal{B}(1^\lambda)]\bigg| \leq \negl(\lambda).  
    \end{align*}

\subsection{Quantum Information}
\begin{lemma}[Almost As Good As New Lemma \cite{aarlemma}, verbatim]\label{lem:gentlemes}
    Let $\rho$ be a mixed state acting on $\C^{d}$. Let $U$ be a unitary and $(\Pi_0, \Pi_1 = I - \Pi_0)$ be projectors all acting on $\C^{d} \otimes \C^{d'}$. We interpret $(U, \Pi_0, \Pi_1)$ as a measurement performed by appending an ancillary system of dimension $d'$ in the state $\ketbra{0}{0}$, applying $U$ and then performing the projective measurement $\Pi_0, \Pi_1$ on the larger system. Assuming that the outcome corresponding to $\Pi_0$ has probability $1 - \epsilon$, we have
    \begin{equation*}
        \trd{\rho}{\rho'} \leq \sqrt{\epsilon}
    \end{equation*}
    where $\rho'$ is the state after performing the measurement, undoing the unitary $U$ and tracing out the ancillary system.
\end{lemma}

\begin{theorem}[\cite{TCC:CakGoy24}]\label{lem:simulproj}
    Let $\rho$ be a bipartite state and $\Lambda = \{\Pi_1, \dots \}, \Lambda' = \{\Pi'_1, \dots \}$ be two projective measurements over each of these registers, respectively. Suppose \begin{equation*}
        \Tr{\Pi_1\otimes \Pi'_1 \rho} \geq 1 - \epsilon.
    \end{equation*}
    Let $M = \{M_i\}_{i \in \mathcal{I}}$ be a general measurement over the first register and fix any $i \in \mathcal{I}$. Let $\tau$ denote the post-measurement state of the second register after applying the measurement $M$ on the first register of $\rho$ and conditioned on obtaining outcome $i$. Let $p_i$ denote probability of outcome $i$, that is $p_i = \Tr{(M_i \otimes I)\rho(M_i^\dagger \otimes I)}$. Then, \begin{equation*}
                \Tr{\Pi'_1 \tau} \geq 1 - \frac{3\sqrt{\epsilon}}{2p_i}.
    \end{equation*}
\end{theorem}

\begin{theorem}[Implementation Independence of Measurements on Bipartite States (\cite{TCC:CakGoy24})]\label{thm:impindep}
    Let $\Lambda = \{M_i\}_{i \in \mathcal{I}}, \Lambda' = \{E_i\}_{i \in \mathcal{I}}$ be two general measurements whose POVMs are equivalent, that is, $M_i^{\dagger}M_i = E_i^{\dagger}E_i$ for all $i \in \mathcal{I}$.

    Let $\rho$ be any bipartite state whose first register has the appropriate dimension for $\Lambda, \Lambda'$. Then, the post-measurement state of the second register conditioned on any outcome $i \in \mathcal{I}$ is the same when either $\Lambda$ or $\Lambda'$ is applied to the first register of $\rho$. That is, \begin{equation*}
        (\Tr\otimes I)\frac{(M_i\otimes I)\rho(M_i^\dagger\otimes I)}{\Tr{(M_i\otimes I)\rho(M_i^\dagger\otimes I)}} =  (\Tr\otimes I)\frac{(E_i\otimes I)\rho(E_i^\dagger\otimes I)}{\Tr{(E_i\otimes I)\rho(E_i^\dagger\otimes I)}}
    \end{equation*}
\end{theorem}

\section{Definitional Work}
\subsection{Malleable-Puncturable Schemes}
We first recall the following template from \cite{ccakan2025copy} that captures most cryptographic schemes along with their security games.
\begin{definition}[Cryptographic Scheme with Security Game (\cite{ccakan2025copy})]\label{def:func}
A cryptographic scheme with security game is a tuple of efficient algorithms $(\ceval, \Ver, \chal, \sampchal)$, defined as follows.
\begin{itemize}
\item $\ceval$ is a (possibly quantum) algorithm that receives as input some key $k \in \mathcal{K}$ and some input $z \in \mathcal{Z}$, and outputs some $y \in \mathcal{Y}$,
    \item $\chal$ is an (possibly quantum) interactive algorithm that represents the setup of the security game such that at the end of the interaction with an adversary, it outputs a key $k \in \mathcal{K}$ and challenger state\footnote{This represents the information needed to sample a challenge. For example, this could be the challenge messages $m_0, m_1$ in a CPA security game for encryption.} $st$,
    \item $\sampchal$ is a classical input-output (possibly quantum) sampler that receives as input $st$, and outputs a challenge $ch$ (which will be sent to the adversary) and the answer key\footnote{This key  here represents the information that is needed to verify adversary's answer. For example, in a CPA security game, this will be the challenge bit $b$ where the challenge ciphertext $ct$ is an encryption of $m_b$.} $ak$,
    \item $\Ver$ is a (possibly quantum) verifier that receives as input the answer key $ak$ and an alleged answer $ans'$, and it outputs $\outtrue$ or $\outfalse$.
\end{itemize}

We define meaningfulness and security as below.

\paragraph{\underline{$\ptriv$-Security:}} For any QPT adversary $\adve$, we have $\Pr[\secgamewp{\adve} = 1] \leq \ptriv(\lambda) + \negl(\lambda)$ where the security game $\secgame$ is defined as follows\footnote{Here, $\ptriv$ denotes the baseline/trivial success probability. For example, for a chosen plaintext security game, we will have $\ptriv = 1/2$.}.
\paragraph{\underline{$\secgamewp{\adve}$}}
\begin{enumerate}
    \item $\chal(1^\lambda)$ and $\adve(1^\lambda)$ interact\footnote{During this interaction, $\adve$ may learn some information related to the secret $s$. For example, in a public-key encryption security game, $\adve$ will learn the public-key}, $\chal$ outputs $k$ and $st$.
    \item The challenger runs $ak,ch \samp \sampchal(st)$.
    \item The adversary $\adve$ receives $ch$ and outputs $ans'$.
    \item The challenger outputs $1$ if $\Ver(ak, ans') = \outtrue$, otherwise outputs $0$.
\end{enumerate}

\paragraph{\underline{Meaningfulness}:} There exists an efficient (possibly quantum) algorithm $\mathcal{B}$ such that $\Ver(ak, \mathcal{B}(k, ch)) = \outtrue$ with probability $1 - \negl(\lambda).$ where the values are sampled as in the experiment above.
\end{definition}

We now define the notion of malleable-puncturable schemes.

\begin{definition}[Malleable-Puncturable Scheme]\label{def:punc}
A malleable-puncturable scheme is a cryptographic scheme $\cscheme = (\ceval, \Ver, \chal, \sampchal)$ with the following additions and modifications.

\begin{itemize}
    \item The key $k$ is parsed as $(C, aux, \cpred)$ where $C$ is a classical circuit.
    \item $\ceval$ algorithm $\ceval(k, z)$ is equivalent to $\ceval^{C}(aux, z)$ ($C$ can be called possibly in superposition).
    \item $\sampchal$ consists of two steps: $\sampinp(st)$ which outputs $x$, and $\sampchalfrominp(st, x)$ which outputs  $ak, ch$.
    \item There is an additional (possibly quantum) efficient algorithm $\fsplit(st, x)$ which outputs $\csplit$ such that for any $x'$: if $\cpred(x') \neq x$ then $\csplit(x') = C(x)$, where $x \samp \sampinp(st).$
 \end{itemize}

\paragraph{\underline{Malleable-Puncturable $\ptriv$-Security:}}
For any QPT adversary $\adve$, we have $\Pr[\puncgamewp{\adve} = 1] \leq \ptriv(\lambda) + \negl(\lambda)$ where the security game $\puncgame$ is defined as follows.
\paragraph{\underline{$\puncgamewp{\adve}$}}
\begin{enumerate}
    \item \textbf{Setup Phase:} $\chal$ and $\adve$ interact, $\chal$ outputs $k = (C, aux, \cpred)$.
    \item The challenger runs $x \samp \sampinp(st)$.
    \item The challenger runs $ak, ch \samp \sampchalfrominp(st, x)$.
    \item The challenger runs $\csplit \samp \fsplit(st, x)$.
    \item  \textbf{Challenge Phase:} The adversary $\adve$ receives $(\csplit, \cpred), aux, ch, x$.
    \item The challenger outputs $1$ if $\Ver(ak, ans') = \outtrue$, otherwise outputs $0$.
\end{enumerate}

\paragraph{\underline{Weak Challenge Resampling Correctness:}}\footnote{Note that this is trivially true for public-key primitives} Any malleable-punctured key $\csplit, \cpred, aux$ can be used to sample from the challenge distribution $\sampchal(st)$.
\end{definition}

Finally, we also require that the distribution $\sampinp(st)$ has min-entropy $\lambda^c$ for some $c \in \R^+$, given the state of the adversary $\adve$ and $st$.

\subsection{Quantum Protection Definitions}\label{sec:qprotect}
We recall copy-protection security, introduced by \cite{Aar09}, which is also referred to as \emph{anti-piracy security} or \emph{unclonability security}. We use the particular formalization given by \cite{ccakan2025copy}.

\begin{definition}[Copy-Protection]\label{def:normal}
      A quantum protection scheme $\copyprot$ is said to satisfy copy-protection security for a cryptographic scheme $\cscheme = (\ceval, \Ver, \chal, \sampchal)$ (\cref{def:func}) if for any QPT adversary $\adve = (\adve_0, \adve_1 \adve_2)$, we have $\Pr[\cprotectgamewp = 1] \leq \ptriv(\lambda) + \negl(\lambda)$ where the security game $\cprotectgamewp$ is defined as follows.

\paragraph{\underline{$\cprotectgamewp$}}
\begin{enumerate}
\item Sample $pp, \regi' \samp \copyprot.\genstate(1^\lambda)$.
    \item $\chal(1^\lambda)$ and $\adve(1^\lambda)$ interact, $\chal$ outputs $k$ and $st$.
    \item Sample $cp \samp \copyprot.\prot(pp, k)$.
    \item Submit  $cp, \regi'$ to $\adve_0$.
    \item The adversary $\adve_0$ outputs two registers $\regi_1, \regi_2$.
\item The challenger runs $ak_1, ch_1 \samp \sampchal(st)$ and $ak_2, ch_2 \samp \sampchal(st)$.
\item $\adve_1$ receives $ch_1, \regi_1$, outputs $ans'_1$.
\item $\adve_2$ receives $ch_2, \regi_2$, outputs $ans'_2$.
    \item $b_{1,Ver} \samp \Ver(ak_1, ans'_1)$ and $b_{2,Ver} \samp \Ver(ak_2, ans'_2)$.
    \item The challenger outputs $1$ if $b_{1,Ver} = 1 \wedge b_{2,Ver}=1$, otherwise outputs $0$.
\end{enumerate}
\end{definition}
We also recall strong anti-piracy, which uses \emph{threshold implementations}. We refer the reader to \cite{C:ALLZZ21} and \cite{TCC:CakGoy24} for definition of threshold implementations. When we write $\tip{\ptriv + \gamma}{(\Ver, \mathcal{D})}$ for some distribution $\mathcal{D}$ (such as $\sampchal(st)$), we mean the threshold implementation for the following mixture of binary projective measurements: Sample $ak, ch \samp \mathcal{D}$ and apply the universal quantum circuit to execute the encoded circuit in $\regi$ on input $ch$ to obtain an outcome $ans'$, and run $\Ver(ak, ans')$ to obtain a bit. When it is clear from context, we will omit writing $\Ver$ and just denote the challenge distribution.
\begin{definition}[Strong Anti-Piracy]\label{def:strongap}
Let $\gamma$ be an inverse polynomial.
    A quantum protection scheme $\copyprot$ is said to satisfy strong copy-protection security if for any QPT adversary $\adve = \adve_0$, we have $\Pr[\strongcprotectgamewp = 1] \leq  \negl(\lambda)$ where the security game $\strongcprotectgamewp$ is defined as follows.

\paragraph{\underline{$\strongcprotectgamewp$}}
\begin{enumerate}
\item Sample $pp, \regi' \samp \copyprot.\genstate(1^\lambda)$.
    \item $\chal(1^\lambda)$ and $\adve(1^\lambda)$ interact, $\chal$ outputs $k$ and $st$.
    \item Sample $cp \samp \copyprot.\prot(pp, k)$.
    \item Submit  $cp, \regi'$ to $\adve_0$.
    \item The adversary $\adve_0$ outputs two registers $\regi_1, \regi_2$.
\item Apply the threshold implementation $\tip{\ptriv + \gamma}{(\Ver, \sampchal(st))}\otimes \tip{\ptriv + \gamma}{(\Ver,\sampchal(st))}$ to $\regi_1, \regi_2$.
    \item The challenger outputs $1$ if both measurements output $1$, otherwise outputs $0$.
\end{enumerate}
\end{definition}
As shown by \cite{C:CLLZ21}, a scheme that satisfies strong anti-piracy for any inverse polynomial $\gamma$ also satisfies regular anti-piracy (\cref{def:normal}).

\section{Technical Tools}
In this section, we recall some technical tools, mostly verbatim from \cite{ccakan2025copy}.

\subsection{Quantum Protection Properties of Coset States}
We first give the following lemma, which essentially helps separate the computational and the information-theoretic quantum protection properties of coset states.
\begin{lemma}\label{lem:shobased}
Consider the following (parametrized) game $\game_\adve^b(1^\lambda)$ (for $c \in \zo$) between a challenger and an adversary $\adve$.

\paragraph{\underline{$\game_\adve^c(1^\lambda)$}}
\begin{enumerate}
    \item The challenger samples a $\F_2^{d(\lambda)}$-coset as follows: Sample a subspace $A \leq \F_2^{d(\lambda)}$ and two elements $a_1, a_2 \in \F_2^{d(\lambda)}$ uniformly at random.
    \item The challenger initialize the register $\regi$ with the coset state $\ket{A_{a_1, a_2}} = \sum_{v \in A} (-1)^{\langle v, a_2\rangle} \ket{v + a_1}$.
    \item The challenger samples a superspace $B_1$ of $A$ of dimension $\frac{3d}{4}$, and a subspace $B_2$ of $A$ of dimension $\frac{d}{4}$ and elements  $z_1 \samp B_1$, $z_2 \samp B_2^\perp$. The challenger also sets $t=z_1+a_1$ and $t' = z_2+a_2$. 

\item The challenger samples $\obfd{M} \samp \io(M^c)$.
\begin{mdframed}
    {\bf \underline{$M^0(b, v)$}}
    
     {\bf Hardcoded: $A, a_1, a_2$}
    \begin{enumerate}[label=\arabic*.]
        \item If $b = 0$: Check if $v \in A + a_1$. If so, output $\outtrue$, otherwise output $\outfalse$, and terminate.
        \item If $b = 1$: Check if $v \in A^\perp + a_2$. If so, output $\outtrue$, otherwise output $\outfalse$.
    \end{enumerate}
    \end{mdframed}
    \begin{mdframed}
    {\bf \underline{$M^1(b, v)$}}
    
     {\bf Hardcoded: ${t, t', B_1,B_2}$}
    \begin{enumerate}[label=\arabic*.]
        \item If $b = 0$: Check if $v \in {B_1 + t}$. If so, output $\outtrue$, otherwise output $\outfalse$, and terminate.
        \item If $b = 1$: Check if $v \in {B_2^\perp + t'}$. If so, output $\outtrue$, otherwise output $\outfalse$.
    \end{enumerate}
    \end{mdframed}
    
    \item The challenger sends $\regi, A, B_1, B_2, t, t'$ to the adversary $\adve$.
    \item The adversary outputs a register $\regis{out}$.
    \item The challenger outputs $\regis{out}, A, a_1, a_2$.
\end{enumerate}
    
   For any QPT adversary $\adve$,
    $\trd{\game_\adve^0(1^\lambda)}{\game_\adve^1(1^\lambda)} \leq \negl(\lambda),$
    assuming the existence of indistinguishability obfuscation and injective one-way functions. 

    Assuming the existence of subexponentially secure indistinguishability obfuscation and injective one-way functions, there exists a constant\footnote{Note that this is a fixed constant that cannot be improved by assuming higher levels of security for iO or OWF.} $c > 0$ such that that for any quantum adversary that runs in time $2^{-(d(\lambda))^c}$,
    $\trd{\game_\adve^0(1^\lambda)}{\game_\adve^1(1^\lambda)} \leq 2^{-(d(\lambda))^c}.$
\end{lemma}

Now we recall the following information-theoretic copy-protection security for coset states.
\begin{lemma}[Monogamy-of-Entanglement for Coset States (\cite{C:CLLZ21}, implicit)]\label{lem:moelast}
Consider the following game between a tuple of adversaries $\adve = (\adve_0, \adve_1, \adve_2)$ and a challenger.

\paragraph{\underline{$\gamewp{\adve}$}}
\begin{enumerate}
    \item The challenger samples a $\F_2^{d(\lambda)}$-coset as follows: Sample a subspace $A \leq \F_2^{d(\lambda)}$ and two elements $a_1, a_2 \in \F_2^{d(\lambda)}$ uniformly at random.
    \item The challenger initialize the register $\regi$ with the coset state $\ket{A_{a_1, a_2}} = \sum_{v \in A} (-1)^{\langle v, a_2\rangle} \ket{v + a_1}$.
    \item The challenger samples a superspace $B_1$ of $A$ of dimension $\frac{3d}{4}$, and a subspace $B_2$ of $A$ of dimension $\frac{d}{4}$ and elements  $z_1 \samp B_1$, $z_2 \samp B_2^\perp$. The challenger also sets $t=z_1+a_1$ and $t' = z_2+a_2$. 
    \item The challenger sends $\regi, B_1, B_2, t, t'$ to the adversary $\adve_0$.
    \item The adversary $\adve_0$ outputs two registers $\regi_1, \regi_2$.
    \item The adversaries\footnote{The adversaries not communicating, but they are entangled through $\regi_1, \regi_2$} $\adve_1, \adve_2$ receive $(A, \regi_1)$ and $(A, \regi_2)$, respectively.
    \item The adversaries $\adve_1, \adve_2$ outputs $v$ and $w$, respectively.
    \item The challenger checks if $v = \canonvec{A}{a_1}$ and $w = \canonvec{A^\perp}{a_2}$. If so, it outputs $1$. Otherwise, it outputs $0$.
\end{enumerate}

For any (unbounded) adversary $\adve$,
\begin{equation*}
    \Pr[\gamewp{\adve}=1] \leq 2^{-\cmoe\cdot d(\lambda)}
\end{equation*}
\end{lemma}

\subsection{Quantum-Proof Reconstructive Extractors}\label{sec:ext}
We recall the notion of \emph{quantum-proof} extractors, which is a randomness extractor that is secure against quantum side information on the source.
\begin{definition}[\protect{Quantum-proof seeded extractor \cite{anindya}}]\label{def:ext}
    A function $\ext: \zo^n \times \zo^d \to \zo^m$ is said to be a \emph{$(k, \epsilon)$-strong quantum-proof seeded extractor} if for any $cq$-state $\rho \in \zo^n \otimes \mathcal{Y}$ of the registers $X, Y$ with $H_{\mathsf{min}}(X | Y) \geq k$, we have
    \begin{equation*}
        \ext(X, S), Y, S \approx_\epsilon U_m, Y, S 
    \end{equation*}
    where $S \leftarrow \zo^d$ and $U_m \samp \zo^m$.
\end{definition}
We call an extractor \emph{reconstructive extractor} if given a distinguisher between $\ext(X, S), Y, S$ versus $U_m, Y, S$ with advantage better than $\epsilon$, there exists an (unbounded) algorithm $\mathcal{E}$ such that $\Pr[\mathcal{E}(Y) = X] \geq 2^{-k}$.

\subsection{Projective and Threshold Implementations}

We recall the following properties of projective and threshold implementations (\cite{TCC:Zhandry20, C:ALLZZ21}). We refer the reader to \cite{C:ALLZZ21} and \cite{TCC:CakGoy24} for definitions of projective and threshold implementations.

\begin{theorem}[\cite{C:ALLZZ21}]\label{thm:singleatiprop}
    Consider the approximate threshold implementation $\atip{\eta}{\delta}{(\mathcal{P}, \mathcal{D})}{\eps}$, associated with a collection of binary projective measurements $\mathcal{P}$, a distribution $\mathcal{D}$ over the index set of $\mathcal{P}$ and a threshold value $\eta \in [0, 1]$, applied to a state $\rho$.

We then have the following.
\begin{itemize}
\item For any state $\rho$, 
\begin{equation*}
   \Tr[\atip{\eta-\eps}{\delta}{(\mathcal{P}, \mathcal{D})}{\eps}\cdot\rho =1] \geq \Tr[\tip{\eta}{(\mathcal{P},\mathsf{D})}\rho = 1] - \delta.
\end{equation*}
\item For any state $\rho$, 
\begin{equation*}
   \Tr[\tip{\eta-\eps}{(\mathcal{P},\mathsf{D})}\rho = 1] \geq \Tr[\atip{\eta}{\delta}{(\mathcal{P}, \mathcal{D})}{\eps}\rho=1] - \delta.
\end{equation*}
 \item $\atip{\eta-\eps}{\delta}{(\mathcal{P}, \mathcal{D})}{\eps}$ is efficient.
    \item $\tip{\eta}{(\mathcal{P},\mathsf{D})}$ is a projection and the collapsed state conditioned on outcome $1$ is a mixture of eigenvectors of $\mathcal{P}$ with eigenvalue $\geq \eta$.
\end{itemize}
\end{theorem}
We give some further generalizations below.
\begin{theorem}[\cite{TCC:CakGoy24}]\label{thm:multiatiprop}
 For any $k \in \mathbb{N}$, let $\mathcal{P}_\ell, \mathcal{D}_\ell$ be a collection of projective measurements and a distribution on the index set of this collection, respectively, and $\eta_\ell \in [0, 1]$ be threshold values for all $\ell \in [k]$. Write $\Tr[\left(\bigotimes_{\ell \in [k]} \mathsf{ATI}^{\eps,\delta}_{\mathcal{P}_\ell,\mathcal{D}_\ell,\eta_\ell}\right)\cdot \rho]$ to denote the probability that the outcome of the joint measurement $\bigotimes_{\ell \in [k]} \mathsf{ATI}^{\eps,\delta}_{\mathcal{P}_\ell,\mathcal{D}_\ell,\eta_\ell}$ applied on $\rho$ is all $1$, and similarly for $\mathsf{TI}$.

 Then, we have the following.
\begin{itemize}

\item    For any $k$-partite state $\rho$, 
\begin{equation*}
   \Tr[\left(\bigotimes_{\ell \in [k]} \mathsf{ATI}^{\eps,\delta}_{\mathcal{P}_\ell, \mathcal{D}_\ell, \eta_\ell - \eps}\right)\rho] \geq \Tr[\left(\bigotimes_{\ell \in [k]} \mathsf{TI}_{\eta_\ell}({\mathcal{P}_\ell}_{\mathcal{D}_\ell})\right)\rho] - k\cdot \delta.
\end{equation*}
\item  For any $k$-partite state $\rho$, let $\rho'$ be the collapsed state obtained after applying $\bigotimes_{\ell \in [k]} \mathsf{ATI}^{\eps,\delta}_{\mathcal{P}_\ell,\mathcal{D}_\ell,\eta_\ell}$ to $\rho$ and obtaining the outcome $1$. Then, 
\begin{equation*}
   \Tr[\left(\bigotimes_{\ell \in [k]}\mathsf{TI}_{\eta_\ell - 2\eps}({\mathcal{P}_\ell}_{\mathcal{D}_\ell})\right)\rho'] \geq 1 - 2k\cdot \delta.
\end{equation*}

\item  For any $k$-partite state $\rho$, let $\rho'$ be the collapsed state obtained after applying $\bigotimes_{\ell \in [k]} \mathsf{ATI}^{\eps,\delta}_{\mathcal{P}_\ell, \mathcal{D}_\ell,\eta}$ to $\rho$ and obtaining the outcome $1$. Then, 
\begin{equation*}
   \Tr[\left(\bigotimes_{\ell \in [k]} \mathsf{ATI}^{\eps,\delta}_{\mathcal{P}_\ell, \mathcal{D}_\ell, \eta_\ell - 3\eps}\right)\cdot\rho'] \geq 1 - 3k\cdot \delta.
\end{equation*}

\item  For any $k$-partite state $\rho$, 
\begin{equation*}
  \Tr[\left(\bigotimes_{\ell \in [k]} \mathsf{TI}_{\eta_\ell - \eps}({\mathcal{P}_\ell}_{\mathcal{D}_\ell})\right)\rho] \geq  \Tr[\left(\bigotimes_{\ell \in [k]} \mathsf{ATI}^{\eps,\delta}_{\mathcal{P}_\ell, \mathcal{D}_\ell, \eta_\ell}\right)\rho] - k\cdot \delta.
\end{equation*}
   
\end{itemize}
\end{theorem}

We now give the following lemma about simulating approximate threshold implementations.
\begin{lemma}[Simulated $\ati$ (\cite{ccakan2025copy})]\label{lem:simati}
There exists a polynomial $q(\cdot)$ and a (parametrized) measurement $\simatiwp{\delta}{\epsilon}{P}{\mathcal{D}}{\eta}{\alpha}$ that receives as input a list of $\ell$ strings (from the support of $\mathcal{D}$) and runs in time $\poly(\ell)$ such that for any $0 < \epsilon, \delta, \eta, \alpha< 1$, for any collection of binary projective measurements $P=\{P_i, I - P_I\}_{i \in \mathcal{I}}$, for any possibly correlated distribution of states $\rho$ with distributions $\mathcal{D}$ over $\mathcal{I}$,
\begin{align*}
    &\Pr[\simatiwp{\delta}{\epsilon}{P}{\mathcal{D}}{\eta-5\epsilon}{\alpha}(s_1, \dots, s_\ell)(\rho) = 1 : s_1, \dots, s_\ell \samp \mathcal{D}] \geq \Pr[\atip{\eta}{\delta}{(P, \mathcal{D})}{\epsilon}(\rho) = 1] - \alpha - 4\delta. \\
    &\Pr[\atip{\eta-5\epsilon}{\delta}{(P, \mathcal{D})}{\epsilon}(\rho) = 1] \geq \Pr[\simatiwp{\delta}{\epsilon}{P}{\mathcal{D}}{\eta}{\alpha}(s_1, \dots, s_\ell)(\rho) = 1 : s_1, \dots, s_\ell \samp \mathcal{D}] - \alpha - 4\delta\\
    &\Pr[\simatiwp{\delta}{\epsilon}{P}{\mathcal{D}}{\eta-5\epsilon}{\alpha}(s_1, \dots, s_\ell)(\rho) = 1 : s_1, \dots, s_\ell \samp \mathcal{D}] \geq \Pr[\tip{\eta}{(P, \mathcal{D})}(\rho) = 1] - \alpha - 4\delta \\
    &\Pr[\tip{\eta-5\epsilon}{(P, \mathcal{D})}(\rho) = 1] \geq \Pr[\simatiwp{\delta}{\epsilon}{P}{\mathcal{D}}{\eta}{\alpha}(s_1, \dots, s_\ell)(\rho) = 1 : s_1, \dots, s_\ell \samp \mathcal{D}] - \alpha - 4\delta
\end{align*}
where $\ell = q(1/\epsilon, \log(1/\delta), 1/\alpha)$. Note that the runtime of $\simati$ depends on the parameters of $\mathcal{D}$ (e.g. sampling runtime of $\mathcal{D}$), but otherwise it does not use oracles access to $\mathcal{D}$ nor does it use any samples from $\mathcal{D}$ (except for $s_1, \dots, s_\ell$ which are given explicitly). 
\end{lemma}

\section{Quantum Protection Construction}\label{sec:cons}
In this section, we give our quantum protection construction $\copyprot$, which is the same as the construction of \cite{ccakan2025copy}, which in turn is essentially the same as the construction of \cite{C:CLLZ21} (when instantiated with iO and coset states rather than ideal oracles and subspace states).

Let $\cscheme = (\ceval, \Ver, \chal, \sampchal)$ be the malleable-puncturable scheme to be protected. Let $\prf$ be a pseudorandom function family scheme and let $\io$ be an indistinguishability obfuscation scheme, with security levels and input-output sizes to be defined in the proof. Let $d(\lambda)$ be a parameter (set in the proof). Further, all the obfuscated programs are implicitly padded to the appropriate sizes as needed in the proof.

\paragraph{\underline{$\copyprot.\genstate(1^\lambda)$}}
\begin{enumerate}
\item Sample a $\F_2^{d(\lambda)}$-coset as follows: Sample a subspace $A \leq \F_2^{d(\lambda)}$ and two elements $a_1, a_2 \in \F_2^{d(\lambda)}$ uniformly at random.
\item Initialize the register $\regi$ with the coset state $\ket{A_{a_1, a_2}} = \sum_{v \in A} (-1)^{\langle v, a_2\rangle} \ket{v + a_1}$.
\item Sample $\obfd{M} \samp \io(M)$ where $M$ is the membership checking program for the cosets as defined below.
 \begin{mdframed}
    {\bf \underline{$M(b, v)$}}
    
     {\bf Hardcoded: $A, a_1, a_2$}
    \begin{enumerate}[label=\arabic*.]
        \item If $b = 0$: Check if $v \in A + a_1$. If so, output $\outtrue$, otherwise output $\outfalse$, and terminate.
        \item If $b = 1$: Check if $v \in A^\perp + a_2$. If so, output $\outtrue$, otherwise output $\outfalse$.
    \end{enumerate}
    \end{mdframed}
\item Output $(pp = \obfd{M}, \regi)$.
\end{enumerate}

\paragraph{\underline{$\copyprot.\prot(pp, \funckey)$}}
\begin{enumerate}
\item Parse $\obfd{M} = pp$.
\item Parse $(C, aux, \cpred) = k$.
\item Sample a PRF key as $K \samp \prf.\keygen(1^\lambda)$.
    \item Sample $\obfd{P} \samp \io(P)$ where $P$ is the following program.
     \begin{mdframed}
    {\bf \underline{$P(x, b, v)$}}
    
     {\bf Hardcoded: $K, \obfd{M}, C$}
    \begin{enumerate}[label=\arabic*.]
        \item Check if $\obfd{M}(b, v) = \outtrue$. If not, output $\bot$ and terminate.

        \item If $b = 0$, output $C(x) \oplus \prf.\ceval(K, x)$ and terminate.
        \item If $b = 1$, output $\prf.\ceval(K, x)$.
    \end{enumerate}
    \end{mdframed}
    \item Output $(\obfd{P}, aux)$.
\end{enumerate}

\paragraph{\underline{$\copyprot.\qpeval(\regi, x)$}}
\begin{enumerate}
\item Parse $((\obfd{P}, aux), \regi') = \regi$.
\item Run $\cscheme.\ceval(aux, z)$ by implementing each circuit-query as follows: Run $\obfd{P}$ on $x, 0, \regi'$ coherently, let $y_0$ denote the outcome. Then, rewind the register $\regi'$ (as in \cref{lem:gentlemes}). Then run $\obfd{P}$ on $x, 1, \regi'$ coherently, let $y_1$ denote the outcome. Then, rewind the register $\regi'$ (as in \cref{lem:gentlemes}). Pass $y_0 \oplus y_1$ as the output.
\end{enumerate}

\begin{theorem}\label{thm:strongcopyprotect}
   For any malleable-puncturable scheme $\cscheme$, $\copyprot$ satisfies strong anti-piracy (\cref{def:strongap}) for all inverse polynomial $\gamma(\lambda)$.
\end{theorem}

\begin{corollary}\label{thm:main}
    Assuming subexponentially secure indistinguishability obfuscation and one-way functions, $\copyprot$ is a secure copy-protection for any malleable-puncturable scheme.
\end{corollary}
\begin{proof}
All the assumed primitives (in the construction and in the proof) can be instantiated using subexponentially secure indistinguishability obfuscation and one-way functions for any subexponential function. Correctness of the scheme is straightforward. The security follows from \cref{thm:strongcopyprotect} and the fact that strong anti-piracy for all inverse polynomial values implies regular anti-piracy (\cref{def:normal}), as shown by \cite{C:CLLZ21}.
\end{proof}

\section{Proof of Copy-Protection Security}
In this section, we prove \cref{thm:strongcopyprotect}. We prove security through a sequence of hybrids, each of which is constructed by modifying the previous one.

Let $\cscheme = (\ceval, \Ver, \chal, \sampinp, \sampchalfrominp, \fsplit)$ be the malleable-puncturable scheme being copy-protected. Suppose for a contradiction that there exists a QPT adversary $\adve$ and polynomials $q_1(\lambda), q_2(\lambda)$ such that $\Pr[\strongcprotectgamewp = 1] \geq \frac{1}{q_1(\lambda)}$ for infinitely many values of $\lambda > 0$, where $\gamma(\lambda) = \frac{1}{q_2(\lambda)}$. We also set the parameter $\epsilon^*(\lambda) = \frac{\gamma(\lambda)}{128}$. Another parameter, $\delta^*(\lambda)$, will be chosen to satisfy the parameter constraints given in the proof.Let $\ext$ be a constructive extractor with output length $\lambda^{c'}$ for some $c' \in \R^+$ (any $c'$ works, the precise value will depend on other parameters). against quantum side-information (\cref{sec:ext}) for sources that are $\extprob$ unpredictable sources. Let $\epsilon_{\ext}$ denote the extractor error level. Let $p_1(\lambda)$ denote the input size of the circuits of $\cscheme$, that is, $\mathcal{X} = \zo^{p_1(\lambda)}$. Let $p_5(\lambda)$ denote the output size of the circuits of $\cscheme$. Let $p_2(\lambda)$ denote the maximum of the circuit size of $\cscheme$ and the key length of the PRF scheme $\prf$. 
Let $\alpha_2(\lambda)$ denote the security of $\ace$, let $\alpha_3(\lambda)$ denote iO security, let $\alpha_4(\lambda)$ denote the security level from \cref{lem:shobased}.  $\alpha_1(\lambda)$ is a parameter that will be clear from the proof. Let $G$ be a PRG with appropriate input-output length.

\paragraph{\underline{$\hyb_0$:}} The original security game $\strongcprotectgamewp$. As a notation, we write $b^1_{Test,1}$ and $b^1_{Test,2}$ to denote the measurement outcomes obtained by the challenger by applying the threshold implementations to the registers $\regi_1$ and $\regi_2$, respectively.  We set $b^1_{Test} = b^1_{Test,1} \wedge b^1_{Test,2}$. Thus, outcome of the experiment in this hybrid is $b^1_{Test}$.

\paragraph{\underline{$\hyb_1$:}} Instead of applying the threshold implementation $\tip{\ptriv+\gamma(\lambda)}{\sampchal(st)}$ on the registers $\regi_1, \regi_2$, the challenger now instead applies approximate threshold implementation $\ati^{(1)} = \atip{\ptriv + \frac{126\gamma}{128}}{\delta^*}{\sampchal(st)}{\epsilon^*}$ on the registers $\regi_1$ and $\regi_2$. We still write $b^1_{Test,1}, b^1_{Test,2}$ to denote the measurement outcomes, respectively. Note that the output of the game is $b^1_{Test}$ in this hybrid.

\paragraph{\underline{$\hyb_2$:}} We modify the sampling of $\obfd{M}$ during the execution of $\copyprot.\genstate$. We first sample a superspace $B_1$ of $A$ of dimension $\frac{3d}{4}$, a subspace $B_2$ of $A$ of dimension $\frac{d}{4}$ and elements  $u_1 \samp B_1$, $u_2 \samp B_2^\perp$. Finally, we set $t=u_1+a_1$ and $t' = u_2+a_2$. Now, we sample $\obfd{M}  \samp \io(M')$ where $M'$ is the following program.
\begin{mdframed}
    {\bf \underline{$M'(b, v)$}}
    
     {\bf Hardcoded: $\textcolor{red}{t, t', B_1,B_2}$}
    \begin{enumerate}[label=\arabic*.]
        \item If $b = 0$: Check if $v \in \textcolor{red}{B_1 + t}$. If so, output $\outtrue$, otherwise output $\outfalse$, and terminate.
        \item If $b = 1$: Check if $v \in \textcolor{red}{B_2^\perp + t'}$. If so, output $\outtrue$, otherwise output $\outfalse$.
    \end{enumerate}
    \end{mdframed}

\paragraph{\underline{$\hyb_3$:}} After applying $\ati^{(1)}$ on the registers $\regi_1$ and $\regi_2$, the challenger applies $\ti^{(2)}\otimes \ti^{(2)}$ to the registers $\regi_1, \regi_2$ where $ \tip{\ptriv + \frac{122\gamma}{128}}{\sampchal(st)}$. We write $b^2_{Test,1}, b^2_{Test,2}$ to denote the measurement outcomes, respectively. Finally, the challenger now outputs $b^1_{Test,1}, b^1_{Test,2}, b^2_{Test,1}, b^2_{Test,2}$ as the output of the game. 

\paragraph{\underline{$\hyb_4$:}}  The challenger samples an $\ace$ instance as $sk \samp \ace.\Setup(1^\lambda)$, $ek = \ace.\mathsf{GenEK}(sk, \outfalse\footnote{This represents the circuit that always outputs $\outfalse$, meaning that the encapsulation key will not be punctured at all.})$ and $dk'=\ace.\mathsf{GenDK}(sk, \outtrue\footnote{This represents the circuit that always outputs $\outtrue$, meaning that the encapsulation key will be punctured everywhere.})$. Then it samples $r^*_\mathsf{ct,1}, r^*_\mathsf{ct,2} \samp \zo^{p_4(\lambda)}$ and computes $ct^*_1 = C^* \oplus G(r^*_\mathsf{ct,1}), ct^*_2 = K \oplus G(r^*_\mathsf{ct,2})$ and $a_1^{Can} = \canonvec{A}{a_1}, a_2^{Can} = \canonvec{A^\perp}{a_2}$. 
We also modify the way the challenger samples $\obfd{P}$: Now it samples it as $\obfd{P}\samp \io(P^{(1)})$. 
      \begin{mdframed}
    {\bf \underline{$P^{(1)}(x, b, v)$}}
    
     {\bf Hardcoded: $K, \obfd{M}, C, \textcolor{red}{\cpred, dk', ct^*_1, ct^*_2}$}
    \begin{enumerate}[label=\arabic*.]
        \item Check if $\obfd{M}(b, v) = \outtrue$. If not, output $\bot$ and terminate.

        \item \textcolor{red}{Compute $x' = \cpred(x)$.} 

        \item \textcolor{red}{Compute $pl = \ace.\dec(dk', x')$.  If $pl \neq \fail$, parse $t || A' || se' || z || ind = pl$.} 

        \item \textcolor{red}{If $pl = \fail$ or if $b = 0 \wedge t = 1$ or if $b = 1 \wedge t = 0$,}
        \begin{enumerate}[label=\arabic*.]
        \item If $b = 0$, output $C(x) \oplus \prf.\ceval(K, x)$ and terminate.
        \item If $b = 1$, output $\prf.\ceval(K, x)$ and terminate.
         \end{enumerate}

        \item \textcolor{red}{If $pl \neq \fail$,}
         \begin{enumerate}[label=\arabic*.]
         \item\textcolor{red}{ If $t = 0$, }
         \begin{enumerate}[label=\arabic*.]
         \item \textcolor{red}{Compute $a' = \canonvec{A'}{v}$ and $r' = \ext(se', a') \oplus z$.}
            \item\textcolor{red}{ Compute $C' = ct^*_1 \oplus G(r')$.}
            \item\textcolor{red}{ Output $C'(x) \oplus \prf.\ceval(K, x)$ and terminate.}
         \end{enumerate}

         \item\textcolor{red}{ If $t = 1$,} \begin{enumerate}[label=\arabic*.]
            \item \textcolor{red}{Compute $a' = \canonvec{(A')^\perp}{v}$ and $r' = \ext(se', a') \oplus z$.}
            \item \textcolor{red}{Compute $K' = ct^*_2 \oplus G(r')$.}
            \item \textcolor{red}{Output $\prf.\ceval(K', x)$ and terminate.}
        \end{enumerate}
         \end{enumerate}
    \end{enumerate}
    \end{mdframed}

\paragraph{\underline{$\hyb_5$:}} The challenger now sets $dk'=\ace.\mathsf{GenDK}(sk, \outfalse)$.

\paragraph{\underline{$\hyb_6$:}} 
We define the following distributions where   \begin{itemize}
     \item $ek_1 = \ace.\mathsf{GenEK}(sk, PRE_1)$ and $PRE_1(m)$ is the circuit that outputs $\outtrue$ if the first bit of $m$ is not $1$. That is, $ek_1$ can only encapsulate messages that start with $1$.
    \item $ek_0 = \ace.\mathsf{GenEK}(sk, PRE_0)$ and $PRE_0(m)$ is the circuit that outputs $\outtrue$ if the first bit of $m$ is not $0$. That is, $ek_0$ can only encapsulate messages that start with $0$.
\end{itemize} 
\begin{mdframed}
\underline{$\mathcal{D}^{(2)}_{1}$}

{\bf Hardcoded: $st, A, ek_0$}
\begin{enumerate}
    \item Sample an extractor seed as $se \samp \zo^{\ell_3(\lambda)}$.
    \item {Sample $r_1^{**} \samp \zo^{p_4(\lambda)}$.}
    \item {Set $pl_1 = 0 || A || se || r_1^{**} || \indtag$.}
    \item Sample $x^* \samp \ace.\stegenc(ek_0, pl_1, \codeof{\sampinp(st)})$.
    \item Sample $ak, ch \samp  \sampchalfrominp(st, x^*)$.
         \item Output $ak, ch$.
\end{enumerate}
 \end{mdframed}
\begin{mdframed}
\underline{$\mathcal{D}^{(2)}_{2}$}

{\bf Hardcoded: $cp^*, C^*, A, ek$}
\begin{enumerate}
    \item Sample an extractor seed as $se \samp \zo^{\ell_3(\lambda)}$.
    \item Sample $r_2^{**} \samp \zo^{p_4(\lambda)}$.
    \item {Set $pl_2 = 1 || A || se || r_2^{**} || \indtag$.}
    \item Sample $x^* \samp \ace.\stegenc(ek_1, pl_2, \codeof{\sampinp(st)})$.
    \item Sample $ak, ch \samp  \sampchalfrominp(st, x^*)$.
         \item Output $ak, ch$.
\end{enumerate}
 \end{mdframed}

Finally, instead of applying $\ti^{(2)}\otimes \ti^{(2)}$ to the registers $\regi_1, \regi_2$, the challenger now applies  $\ati^{(3,1)}\otimes \ati^{(3,2)}$ to the registers $\regi_1, \regi_2$, where $\ati^{(3,1)} = \atip{\ptriv + \frac{100\gamma}{128}}{\delta^*}{\mathcal{D}^{(2)}_1}{\epsilon^*}$ and $\ati^{(3,2)} = \atip{\ptriv + \frac{100\gamma}{128}}{\delta^*}{\mathcal{D}^{(2)}_2}{\epsilon^*}$. We still write $b^2_{Test,1}, b^2_{Test,2}$ to denote the measurement outcomes, respectively. Also, the challenger still outputs $b^1_{Test,1}, b^1_{Test,2}, b^2_{Test,1}, b^2_{Test,2}$ as the output of the game.

\paragraph{\underline{$\hyb_7$:}} The challenger flips a random coin $c' \samp \{1,2\}$ and only applies $\ati^{(3,c')}$ to $\regi_{c'}$, instead of testing both registers. Let $b^2_{Test}$  denote the measurement outcome. The challenger now outputs $b^1_{Test,1}, b^1_{Test,2}, b^2_{Test}$ as the experiment outcome.

\paragraph{\underline{$\hyb_8$:}} Instead of applying $\ati^{(3,c')}$ to $\regi_{c'}$, the challenger samples $ak^*, ch^* \samp \mathcal{D}^{(2)}_{c'}$ and runs $b^2_{Ver} \samp U(\regi_{c'}, ch)$.  The challenger now outputs $b^1_{Test,1}, b^1_{Test,2}, b^2_{Ver}$ as the experiment outcome. We write $x^*, r^{**}, se^*, pl^*$ to denote the values sampled/computed during the sampling of $ak^*, ch^*$.

\paragraph{\underline{$\hyb_9$:}}  First, we compute $z^*$ as follows.
If $c' = 1$, we compute $C' = ct_1^* \oplus G(\ext(se^*, a_1^{Can})\oplus r^{**})$ and $z^* = C'(act^*) \oplus F(K, act^*)$ where $C'$ is interpreted as a circuit. If $c' = 2$, we compute $z^* = \prf.\ceval(K', x^*)$ where $K' = ct_2^* \oplus G(\ext(se^*, a_2^{Can})\oplus r^{**})$. 
 The challenger computes $dk'' = \ace.\gendk(sk, \mathds{1}_{pl^*})$.

We modify the way the challenger samples $\obfd{P}$: Now it samples it as $\obfd{P}\samp \io(P^{(2)})$. 
      \begin{mdframed}
    {\bf \underline{$P^{(2)}(x, b, v)$}}
    
     {\bf Hardcoded: $K, \obfd{M}, C, {\cpred, ct^*_1, ct^*_2}, \textcolor{red}{dk'', x^*, z^*}$}
    \begin{enumerate}[label=\arabic*.]
        \item Check if $\obfd{M}(b, v) = \outtrue$. If not, output $\bot$ and terminate.

        \item {Compute $x' = \cpred(x)$.} 

                \item \textcolor{red}{If $x' = x^*$,}
        \begin{enumerate}[label=\arabic*.]
            \item \textcolor{red}{If $b = 0$ and $c' = 1$, output $z^*$ and terminate}.
            \item \textcolor{red}{If $b = 1$ and $c' = 2$, output $z^*$ and terminate}.
            \item \textcolor{red}{Otherwise if $b=1$, output $\prf.\ceval(K, x)$ and terminate.}
            \item \textcolor{red}{Otherwise if $b=0$, output $C(x) \oplus \prf.\ceval(K, x)$ and terminate.}
        \end{enumerate}

        \item {Compute $pl = \ace.\dec(\textcolor{red}{dk''}, x')$.  If $pl \neq \fail$, parse $t || A' || se' || z || ind = pl$.} 

        \item {If $pl = \fail$ or if $b = 0 \wedge t = 1$ or if $b = 1 \wedge t = 0$,}
        \begin{enumerate}[label=\arabic*.]
        \item If $b = 0$, output $C(x) \oplus \prf.\ceval(K, x)$ and terminate.
        \item If $b = 1$, output $\prf.\ceval(K, x)$ and terminate.
         \end{enumerate}

        \item {If $pl \neq \fail$,}
         \begin{enumerate}[label=\arabic*.]
         \item{ If $t = 0$, }
         \begin{enumerate}[label=\arabic*.]
         \item {Compute $a' = \canonvec{A'}{v}$ and $r' = \ext(se', a') \oplus z$.}
            \item{ Compute $C' = ct^*_1 \oplus G(r')$.}
            \item{ Output $C'(x) \oplus \prf.\ceval(K, x)$ and terminate.}
         \end{enumerate}

         \item{ If $t = 1$,} \begin{enumerate}[label=\arabic*.]
            \item{Compute $a' = \canonvec{(A')^\perp}{v}$ and $r' = \ext(se', a') \oplus z$.}
            \item {Compute $K' = ct^*_2 \oplus G(r')$.}
            \item {Output $\prf.\ceval(K', x)$ and terminate.}
        \end{enumerate}
         \end{enumerate}
    \end{enumerate}
    \end{mdframed}

    Also, from this hybrid on, our obfuscated programs will have two branches depending on the value of $c'$. However, since $c'$ is sampled at the beginning of the experiment, we simply (and implicitly) remove the dead branch (i.e. if $c' = 1$, we can remove the lines that are executed if $c' = 1$, along with the hardwired values that are only used there). Note that this does not change the proof thanks to $\io$ security.

\paragraph{\underline{$\hyb_{10}$:}} The challenger now samples $x^* \samp \sampinp(st)$ instead of sampling using $\ace.\stegenc$.

\paragraph{\underline{$\hyb_{11}$:}} The challenger now samples $\obfd{M}$ as $\obfd{M} \samp \io(M)$.

\paragraph{\underline{$\hyb_{12}$:}} We sample $K\{x^*\} \samp \prf.\punc(K, x^*)$. We also set $s_1^* = \prf.\ceval(K, x^*)$ and $s_2^* = C(x^*) \oplus s_1^*$. We also sample $\csplit \samp \fsplit(st, x^*)$. Then, we modify the obfuscated program $\obfd{P}$ by now sampling it as $\obfd{P} \samp \io(P^{(3)})$.

      \begin{mdframed}
    {\bf \underline{$P^{(3)}(x, b, v)$}}
    
     {\bf Hardcoded: $\obfd{M}, {\cpred, ct^*_1, ct^*_2}, dk'', x^*, z^*, \textcolor{red}{\csplit, K\{x^*\}, s_1^*, s_2^*}$}
    \begin{enumerate}[label=\arabic*.]
        \item Check if $\obfd{M}(b, v) = \outtrue$. If not, output $\bot$ and terminate.

        \item {Compute $x' = \cpred(x)$.} 

                \item {If $x' = x^*$,}
        \begin{enumerate}[label=\arabic*.]
            \item{If $b = 0$ and $c' = 1$, output $z^*$ and terminate}.
            \item{If $b = 1$ and $c' = 2$, output $z^*$ and terminate}.
            \item{Otherwise if $b=1$, output \textcolor{red}{$s_1^*$} and terminate.}
            \item{Otherwise if $b=0$, output \textcolor{red}{$s_2^*$} and terminate.}
        \end{enumerate}

        \item {Compute $pl = \ace.\dec({dk''}, x')$.  If $pl \neq \fail$, parse $t || A' || se' || z || ind = pl$.} 

        \item {If $pl = \fail$ or if $b = 0 \wedge t = 1$ or if $b = 1 \wedge t = 0$,}
        \begin{enumerate}[label=\arabic*.]
        \item If $b = 0$, output $\textcolor{red}{\csplit}(x) \oplus \prf.\ceval(\textcolor{red}{K\{x^*\}}, x)$ and terminate.
        \item If $b = 1$, output $\prf.\ceval(\textcolor{red}{K\{x^*\}}, x)$ and terminate.
         \end{enumerate}

        \item {If $pl \neq \fail$,}
         \begin{enumerate}[label=\arabic*.]
         \item{ If $t = 0$, }
         \begin{enumerate}[label=\arabic*.]
         \item {Compute $a' = \canonvec{A'}{v}$ and $r' = \ext(se', a') \oplus z$.}
            \item{ Compute $C' = ct^*_1 \oplus G(r')$.}
            \item{ Output $C'(x) \oplus \prf.\ceval(K, x)$ and terminate.}
         \end{enumerate}

         \item{ If $t = 1$,} \begin{enumerate}[label=\arabic*.]
            \item{Compute $a' = \canonvec{(A')^\perp}{v}$ and $r' = \ext(se', a') \oplus z$.}
            \item {Compute $K' = ct^*_2 \oplus G(r')$.}
            \item {Output $\prf.\ceval(K', x)$ and terminate.}
        \end{enumerate}
         \end{enumerate}
    \end{enumerate}
    \end{mdframed}

\paragraph{\underline{$\hyb_{13}$:}} We modify the way $ct_1^*,ct_2^*$ are computed: Now we sample $ct_1^*, ct_2^* \samp \zo^{p_2(\lambda)}$.

\paragraph{\underline{$\hyb_{14}$:}} If $c' = 1$, we sample $s_1^*$ uniformly at random, and also set $z^* = C'(x^*) \oplus s_1^*$  where $C' = ct^*_1 \oplus G(\ext(se^*, a_1^{Can})\oplus r^{**})$, and also set $s_2^* = \bot$. If $c' = 2$, we sample $s_2^*$ uniformly at random and set $s_1^* = \bot$. Also recall that when $c'=2$, $z^*$ is computed as $z^* = \prf.\ceval(K', x^*)$ where $K' = ct^*_2 \oplus G(\ext(se^*, a_2^{Can})\oplus r^{**})$, which does not use the \emph{actual} PRF key $K$.

\begin{claim}\label{lem:unclonlasthyb}
There exists a polynomial $g(\cdot)$ such that
    $\Pr[b^2_{Ver} = 1] \geq \frac{1}{g(\lambda)}$ where $b^1_{Test,1}, b^1_{Test,2}, b^2_{Ver} \samp \hyb_{14}$.
\end{claim}
We prove this claim in \cref{sec:unclonlasthyb}.

\begin{lemma}
    $\Pr[b^2_{Ver} = 1] \leq \negl(\lambda)$ where $b^1_{Test,1}, b^1_{Test,2}, b^2_{Ver} \samp \hyb_{14}$.
\end{lemma}
\begin{proof}
    This follows directly by the malleable-puncturable security of the scheme $\cscheme$ since the experiment only uses the punctured circuit $\csplit$.
\end{proof}

Since the statements above contradict each other, our hypothesis that there exists a QPT adversary $\adve$ and polynomials $q_1(\lambda), q_2(\lambda)$ such that $\Pr[\strongcprotectgamewp = 1] \geq \frac{1}{q_1(\lambda)}$ for infinitely many values of $\lambda > 0$ is wrong. Thus, for any polynomial $q_2(\lambda)$, we have  $\Pr[\strongcprotectgamewp = 1] \leq \negl(\lambda)$ where $\gamma = \frac{1}{q_2(\lambda)}$, which completes the proof.

\subsection{Indistinguishability of Hybrids: Proof of \cref{lem:unclonlasthyb}}\label{sec:unclonlasthyb}
\begin{claim}
$\Pr[b^1_{Test} = 1] \geq \frac{3}{4q_1(\lambda)}$
    where $b^1_{Test}  \samp \hyb_1$.
\end{claim}
\begin{proof}
Follows from \cref{thm:multiatiprop}.
\end{proof}

\begin{claim}
For $b^1_{Test,1}, b^1_{Test,2} \samp \hyb_2$,
    \begin{itemize}
    \item $\Pr[b^1_{Test,1} = 1 \wedge b^1_{Test, 2} = 1] \geq \frac{1}{2q_1(\lambda)}$
\end{itemize}
\end{claim}
\begin{proof}
    This follows directly from \cref{lem:shobased}.
\end{proof}

\begin{claim}
For $b^1_{Test,1}, b^1_{Test,2}, b^2_{Test,1}, b^2_{Test,2} \samp \hyb_3$,
    \begin{itemize}
    \item $\Pr[b^1_{Test,1} = 1 \wedge b^1_{Test, 2} = 1] \geq \frac{1}{2q_1(\lambda)}$
    \item $\Pr[b^2_{Test,1} \wedge b^2_{Test,2} = 1 \big| b^1_{Test,1} = 1 \wedge b^1_{Test,2} = 1] \geq  1 -  6\delta^*$
\end{itemize}
\end{claim}
\begin{proof}
Follows from \cref{thm:multiatiprop}.
\end{proof}

\begin{claim}
$\hyb_3\approx_{\alpha_3} \hyb_4$
\end{claim}
\begin{proof}
Observe that $P$ and $P^{(1)}$ have exactly the same functionality: Due to the constrained decapsulation safety of $\ace$, $\ace.\dec(dk', x)$ will always output $\bot$ since $dk'= \ace.\mathsf{GenDK}(ek, \outtrue)$ (i.e, the key is punctured everywhere). Hence, the result follows by the security of $\io$.
\end{proof}

\begin{claim}
$\hyb_4\approx_{\alpha_2} \hyb_5$
\end{claim}
\begin{proof}
    Observe that the experiments do not use the encryption key $ek$ at all, and they do not use any $\ace$ ciphertexts either. Thus, the result follows by the puncture hiding security of $\ace$.
\end{proof}

\begin{claim}
For $b^1_{Test,1}, b^1_{Test,2}, b^2_{Test,1}, b^2_{Test,2} \samp \hyb_5$,
    \begin{itemize}
    \item $\Pr[b^1_{Test,1} = 1 \wedge b^1_{Test, 2} = 1] \geq \frac{1}{4q_1(\lambda)}$
    \item $\Pr[b^2_{Test,1} \wedge b^2_{Test,2} = 1 \big| b^1_{Test,1} = 1 \wedge b^1_{Test,2} = 1] \geq  1 -  6\delta^* - 8q_1(\lambda)\cdot(\alpha_3 + \alpha_2)$
\end{itemize}
\end{claim}
\begin{proof}
This follows by combining the claims above.
\end{proof}

\begin{claim}\label{lem:moebased}
\begin{itemize}
    \item $\Pr[b^1_{Test,1} = 1 \wedge b^1_{Test, 2} = 1] \geq \frac{1}{4q_1(\lambda)}$
    \item There exists a polynomial $g_2(\lambda)$ such that
\begin{equation*}
     \Pr[b^2_{Test,1} = 0 \wedge b^2_{Test,2} = 0 \big| b^1_{Test,1} = 1 \wedge b^1_{Test,2} = 1] \leq  1 - \frac{1}{g_2(\lambda)}
\end{equation*}
 where $b^1_{Test,1}, b^1_{Test,2}, b^2_{Test,1}, b^2_{Test,2} \samp \hyb_6$,
\end{itemize}
\end{claim}
\begin{proof}
    We prove this claim in \cref{sec:moebased}.
\end{proof}

\begin{claim}
    For $b^1_{Test,1}, b^1_{Test,2}, b^2_{Test} \samp \hyb_{7}$,
    \begin{itemize}
    \item $\Pr[b^2_{Test} = 1 \wedge b^1_{Test,1} = 1 \wedge b^1_{Test,2} = 1] \geq \frac{1}{8\cdot q_1(\lambda)\cdot g_2(\lambda)}$ 
    \end{itemize}
\end{claim}
\begin{proof}
    This follows directly from the claim above, since with probability $1/2$, $c'$ will \emph{hit} the \emph{good} outcome between $b^2_{Test,1}, b^2_{Test,2}$.
\end{proof}

\begin{claim}
    For $b^1_{Test,1}, b^1_{Test,2}, b^2_{Ver} \samp \hyb_{8}$,
    \begin{itemize}
    \item $\Pr[b^2_{Ver} = 1 \big| b^1_{Test,1} = 1 \wedge b^1_{Test,2} = 1] \geq \ptriv + \frac{90\gamma}{128}$ 
    \end{itemize}
\end{claim}
\begin{proof}
    This follows easily by the properties of $\ati$.
\end{proof}

\begin{lemma}
    $\hyb_8\approx \hyb_9$.
\end{lemma}
\begin{proof}
    We will prove the two cases $c'=1$ and $c'=2$ separately.

First, the case $c'=1$. We claim that $P^{'}$ and $P^{(2)}$ have exactly the same functionality. Note that once we prove this, the result follows by $\io$ security. Now we prove our claim. Observe that $P^{'}$ and $P^{(2)}$ can differ in two places: (i) when using the decryption key $dk'$ in the line $pl = \ace.\dec(dk', x)$, since $P^{'}$ has $dk = \ace.\mathsf{GenDK}(sk, \outfalse)$ but $P^{(2)}$ has $dk' = \ace.\mathsf{GenDK}(sk, \mathds{1}_{pl^*})$, and (ii) if the input satisfies $x = x^*$. First, by the safety of constrained decapsulation of $\ace$, we know that the output of $\ace.\dec$ on $x$ can be different for $dk = \ace.\mathsf{GenDK}(sk, \outfalse)$ versus $dk' = \ace.\mathsf{GenDK}(sk, \mathds{1}_{pl^*})$ only if $x = \ace.\enc(ek, pl^*) = x^*$. Thus, we get that $P^{'}$ and $P^{(2)}$ can possibly differ only on inputs such that $x = x^*$. Further it easy to see that $P^{'}(x,b,v)$ versus $P^{(2)}(x,b,v)$ can only possibly differ if $\obfd{M}(b,v)=\outtrue$ is satisfied too. Thus, we only need to consider inputs $x,b,v$ such that $x = act^*$ and $\obfd{M}(b,v)=\outtrue$. For such inputs, we will further consider two cases: $b = 0$ and $b = 1$. For $b = 1$, it is easy to see that $P^{'}(x,b,v) = \prf.\ceval(K,x^*) = P^{(2)}(x,b,v)$. Finally, we consider the case $b = 0$. Now, $P^{'}(x,b,v)$ enters the case $pl \neq \perp, t = 0$ since $act^*$ decrypts to $pl^*$  by correctness of decapsulation. Here, it gets $a' = a_{1}^{Can}$ since $v \in A+a_1$, and then it gets $r' = \ext(se^*, a_1^{Can})\oplus r^{**}$ since $pl = pl^*$. Then, it gets $C' = ct_1 \oplus G(\ext(se^*, a_1^{Can})\oplus r^{**})$. Thus, we get $P^{'}(x,b,v) = z^*$. Now considering $P^{(2)}(x,b,v)$ when $b = 0$, $x=x^*$ and $v \in A+s$, it is easy to see that we again have $P^{(2,1)}(x,b,v) = z^*$. Thus, $P^{(1)}$ and $P^{(2)}$ have the same output for all $x,b,v$.

Now we prove the case $c'=2$. We claim that $P^{(')}$ and $P^{(2)}$ have exactly the same functionality. Similar to above, once we prove this, the result follows by $\io$ security. The claim that $P^{'}$ and $P^{(2)}$ have the same functionality follows by a case-by-case inspection, similar to above.
\end{proof}

\begin{lemma}
    $\hyb_{9}\approx \hyb_{10}$.
\end{lemma}
\begin{proof}
Observe that these experiments only use the constrained keys $dk'' = \ace.\mathsf{GenDK}(sk, \mathds{1}_{pl^*})$ and $ek' = \ace.\mathsf{GenEK}(sk, \mathds{1}_{pl^*})$. Thus, the result follows by the steganographic ciphertext security of $\ace$.
\end{proof}

\begin{lemma}
    $\hyb_{10}\approx \hyb_{11}$.
\end{lemma}
\begin{proof}
    Follows from \cref{lem:shobased}.
\end{proof}

\begin{lemma}
    $\hyb_{11}\approx \hyb_{12}$.
\end{lemma}
\begin{proof}
    By the malleable-punctured key correctness of the cryptographic scheme $\cscheme$ and the punctured key correctness of the scheme $\prf$, the result follows by the security of $\io$.
\end{proof}

\begin{lemma}
    $\hyb_{12}\approx \hyb_{13}$.
\end{lemma}
\begin{proof}
    Observe that the experiments only uses the PRG output $ct^*_1 = G(r^*_1)$ and $ct^*_2 = G(r^*_2)$, and the random strings $r^*_1, r^*_2$ do not appear anywhere else. Thus, the result follows by the security of $G$.
\end{proof}

\begin{lemma}
    $\hyb_{13} \approx \hyb_{14}$.
\end{lemma}
\begin{proof}
    First, observe that the experiments only use the punctured PRF key $K\{act^*\}$. Thus, by the punctured key security of $\prf$, we can always change $\prf.\ceval(K, act^*)$ to a uniformly random string. This completes the proof of the case $c' = 1$. When $c' = 2$, observe that the experiment does not use $\prf.\ceval(K, act^*)$, thus we can replace $s_2^* = C^*(act^*) \oplus \prf.\ceval(K, act^*)$ with a uniformly random string.
\end{proof}

\subsection{Reduction to Monogamy-of-Entanglement: Proof of \cref{lem:moebased}}\label{sec:moebased}
This part of our proof is mostly the same as \cite{ccakan2025copy}.

The first item, that is $\Pr[b^1_{Test,1} = 1 \wedge b^1_{Test, 2} = 1] \geq \frac{1}{4q_1(\lambda)}$, follows easily by the previous claim (since $b^1_{Test,1} = 1, b^1_{Test, 2}$ have the same marginal distribution in $\hyb_5$ and $\hyb_4$).

Now we prove the second item. Suppose for a contradiction that \begin{equation}\label{eqn:moeimp}
    \Pr[b^2_{Test,1} = 0 \wedge b^2_{Test,2} = 0 \big| b^1_{Test,1} = 1 \wedge b^1_{Test,2} = 1] \geq 1 - \negl(\lambda)
\end{equation}
where $b^1_{Test,1}, b^1_{Test,2}, b^2_{Test,1}, b^2_{Test,2} \samp \hyb_5$.
We will construct an adversary for \cref{lem:moelast}.

Consider the following modified version $\hyb_4^{b}$ of $\hyb_4$: Simulate $\hyb_4$ until right after $\ati^{(1)}\otimes \ati^{(1)}$ is applied to $(\regi_1, \regi_2)$, but do not apply $\ti^{(2)}\otimes \ti^{(2)}$. Now, output $\regi_1, \regi_2, A, ek_b, C^*, cp^*, b^1_{Test,1}, b^1_{Test,2}$ as the outcome of the experiment, where \begin{itemize}
    \item $ek_1 = \ace.\mathsf{GenEK}(sk, PRE_1)$ and $PRE_1(m)$ is the circuit that outputs $\outtrue$ if the first bit of $m$ is not $1$. That is, $ek_1$ can only encapsulate messages that start with $1$.
    \item $ek_0 = \ace.\mathsf{GenEK}(sk, PRE_0)$ and $PRE_0(m)$ is the circuit that outputs $\outtrue$ if the first bit of $m$ is not $0$. That is, $ek_0$ can only encapsulate messages that start with $0$.
\end{itemize}
We write $\hyb_4^{'}$ to denote the version that outputs both $ek_0$ and $ek_1$.

\begin{claim}\label{claim:moeclaimp75}
    Let $\mathcal{M}$ be $\ati^{(3,2)} = \atip{\ptriv + \frac{100\gamma}{128}}{\delta^*}{\mathcal{D}^{(2)}_2}{\epsilon^*}$. Then,

    \begin{equation*}
        \Pr[\ti^{(2)}(\regi_1) = 1 \big|  \mathcal{M}(\regi_2, C^*, cp^*, ek_1) = 0 \wedge b^1_{Test,1} = 1\wedge b^1_{Test,2} = 1] \geq 1 - 3\cdot \sqrt{6\delta^* + 8q_1(\lambda)\cdot(\alpha_3 + \alpha_2)}
    \end{equation*}
    where $\regi_1, \regi_2, A, ek_0, ek_1, C^*, cp^*, b^1_{Test,1}, b^1_{Test,2} \samp \hyb_4'$,

\end{claim}
\begin{proof}
    We already know 
    \begin{equation*}
        \Pr[(\ti^{(2)}\otimes\ti^{(2)})(\regi_1, \regi_2) = (1,1) \big| b^1_{Test,1} = 1\wedge b^1_{Test,2} = 1] \geq 1 -  6\delta^* - 8q_1(\lambda)\cdot(\alpha_3 + \alpha_2).
    \end{equation*}
    Then the result follows by \cref{eqn:moeimp} and \cref{lem:simulproj} since  $\ti^{(2)}$ is a projective measurement.
\end{proof}

\begin{claim}\label{claim:moeclaimp76}
     Let $\mathcal{M}$ be $\ati^{(3,2)} = \atip{\ptriv + \frac{100\gamma}{128}}{\delta^*}{\mathcal{D}^{(2)}_2}{\epsilon^*}$. Then,

    \begin{equation*}
        \Pr[\ati^{(3,1)}(\regi_1) = 0 \big|  \mathcal{M}(\regi_2, C^*, cp^*, ek_1) = 0 \wedge b^1_{Test,1} = 1\wedge b^1_{Test,2} = 1] \geq 1 - \negl(\lambda)
    \end{equation*}
     where $\regi_1, \regi_2, A, ek_0, ek_1, C^*, cp^*, b^1_{Test,1}, b^1_{Test,2} \samp \hyb_4'$.

\end{claim}

\begin{claim}\label{claim:moeclaimp14}
    Let $\mathcal{M}$ be a QPT algorithm with $1$-bit output such that
    \begin{itemize}
        \item $\Pr[\mathcal{M}(\regi_2, C^*, cp^*, ek_1) = 0] \geq \frac{1}{g_3(\lambda)}$ for some polynomial $g_3(\cdot)$
        \item $  \Pr[\ti^{(2)}(\regi_1) = 1 \big|  \mathcal{M}(\regi_2, C^*, cp^*, ek_1) = 0 \wedge b^1_{Test,1} = 1\wedge b^1_{Test,2} = 1] \geq 1 - 3\cdot \sqrt{+6\delta^* +8q_1(\lambda)\cdot(\alpha_3 + \alpha_2)}$,
        \item $\Pr[\ati^{(3,1)}(\regi_1) = 0 \big|  \mathcal{M}(\regi_2, C^*, cp^*, ek_1) = 0 \wedge b^1_{Test,1} = 1\wedge b^1_{Test,2} = 1] \geq 1 - \negl(\lambda)$.
    \end{itemize} 
    
    Then, there exists an (unbounded) algorithm $\mathcal{E}^*_1$ such that
    \begin{equation*}
        \Pr[\mathcal{E}^*_1(\regi_1, A, ek_0, C^*, cp^*) = a_1^{Can} \big| b^1_{Test,1}=1\wedge b^1_{Test,2}=1 \wedge \mathcal{M}(\regi_2) = 0] \geq \extprob
    \end{equation*}
     where $\regi_1, \regi_2, A, ek_0, ek_1, C^*, cp^*, b^1_{Test,1}, b^1_{Test,2} \samp \hyb_4'$.
\end{claim}
We prove this claim in \cref{sec:moeclaimp14}.

\begin{claim}\label{claim:moeclaimp16}
    Let $\mathcal{EV}_1$ be an algorithm with $1$-bit output such that
    \begin{itemize}
    \item $\mathcal{EV}_1$ runs in subexponential time $rt(\lambda)$ (precise value will be clear from the proof)
        \item $\Pr[\mathcal{EV}_1(\regi_1, A, ek_0, C^*, cp^*) = 1] \geq\extprob$
        \item \begin{align*}
            &\Pr[\ti^{(2)}(\regi_2) = 1 \big|  \mathcal{EV}_1(\regi_1, A, ek_0, C^*, cp^*, a_1^{Can}) = 1 \wedge b^1_{Test,1} = 1\wedge b^1_{Test,2} = 1] \geq \\ &1 - \frac{3}{2}\cdot \sqrt{6\delta^* + 8q_1(\lambda)\cdot(\alpha_3 + \alpha_2)}\cdot\frac{2}{\extprob},
        \end{align*}
        \item $\Pr[\ati^{(3,2)}(\regi_2) = 0 \big|  \mathcal{EV}_1(\regi_1, A, ek_0, C^*, cp^*, a_1^{Can}) = 1  \wedge b^1_{Test,1} = 1\wedge b^1_{Test,2} = 1] \geq \frac{1}{4}\extprob$
    \end{itemize} 
    
 Then, there exists an (unbounded) algorithm $\mathcal{E}^*_2$ such that
    \begin{equation*}
        \Pr[\mathcal{E}^*_2(\regi_2, A, ek_1, C^*, cp^*) = a_2^{Can} \big| b^1_{Test,1}=1\wedge b^1_{Test,2}=1 \wedge \mathcal{EV}_1(\regi_1, A, ek_0, C^*, cp^*) = 1] \geq \extprob
    \end{equation*}
     where $\regi_1, \regi_2, A,ek_0, ek_1, C^*, cp^*, b^1_{Test,1}, b^1_{Test,2} \samp \hyb_4'$.
\end{claim}
We prove this claim in \cref{sec:moeclaimp14}.

\begin{claim}
   There exists unbounded algorithms $\mathcal{E}^*_1, \mathcal{E}^*_2$ such that
   \begin{equation*}
       \Pr[\mathcal{E}^*_1(\regi_1, A, ek_0, C^*, cp^*)= a_1^{Can} \wedge \mathcal{E}^*_2(\regi_2, A, ek_1, C^*, cp^*) = a_2^{Can}  \big|b^1_{Test,1}=1\wedge b^2_{Test,2}=1] \geq (\extprob)^2\cdot \frac{1}{2}.
   \end{equation*}
\end{claim}
where $\regi_1, \regi_2, A, ek_0, ek_1, C^*, cp^*, b^1_{Test,1}, b^1_{Test,2} \samp \hyb_4'$.
\begin{proof}
Set $\mathcal{M} = \ati^{(3,2)}$.
First, by \cref{claim:moeclaimp75},\cref{claim:moeclaimp76}, and \cref{claim:moeclaimp14}, we have that  there exist $\mathcal{E}_1^*$ such that \begin{equation*}
        p_1 = \Pr[\mathcal{E}^*_1(\regi_1, A, ek_0, C^*, cp^*) = a_1^{Can} \big| b^1_{Test,1}=1\wedge b^1_{Test,2}=1 \wedge \mathcal{M}(\regi_2, C^*, cp^*, ek_1) = 0] \geq \extprob
    \end{equation*}
    which also implies
    \begin{equation*}
        p_1 = \Pr[\mathcal{E}^*_1(\regi_1, A, ek_0, C^*, cp^*) = a_1^{Can} \big| b^1_{Test,1}=1\wedge b^1_{Test,2}=1] \geq \extprob\cdot\frac{1}{2}
    \end{equation*}

Now, let $\mathcal{EV}_1(1^\lambda)$ be an algorithm simply simulates the output distribution of the equality check $\mathcal{E}^*_1(\regi_1, A, ek_0, C^*, cp^*) =^{?} a_1^{Can}$. That is, $\mathcal{EV}_1(1^\lambda)$ simulates a biased coin and outputs $1$ with probability $p_1$. 
Now we claim the following.
 \begin{itemize}
        \item \begin{align*}
            &\Pr[\ti^{(2)}(\regi_2) = 1 \big|  \mathcal{EV}(\regi_1, A, ek_0, C^*, cp^*, a_1^{Can}) = 1 \wedge b^1_{Test,1} = 1\wedge b^1_{Test,2} = 1] \geq \\ &1 - \frac{3}{2}\cdot \sqrt{  6\delta^* +8q_1(\lambda)\cdot(\alpha_3 + \alpha_2)}\cdot\frac{2}{\extprob},
        \end{align*}
        \item $\Pr[\mathcal{M}(\regi_2, C^*, cp^*, ek_1) = 0 \big| \mathcal{EV}(\regi_1, A, ek_0, C^*, cp^*, a_1^{Can}) = 1  \wedge b^1_{Test,1} = 1\wedge b^1_{Test,2} = 1] \geq \frac{1}{4}\extprob$
    \end{itemize} 

  For the first claim, we already know 
    \begin{equation*}
        \Pr[(\ti^{(2)}\otimes\ti^{(2)})(\regi_1, \regi_2) = (1,1) \big| b^1_{Test,1} = 1\wedge b^1_{Test,2} = 1] \geq 1 -  6\delta^* - 8q_1(\lambda)\cdot(\alpha_3 + \alpha_2).
    \end{equation*}
    Then the first claim follows by \cref{lem:simulproj} since $\Pr[\mathcal{EV}((\regi_1, A, ek_0, C^*, cp^*), a_1^{Can}) = 1 \big| b^1_{Test,1} = 1\wedge b^1_{Test,2} = 1] \geq\extprob   $ and  $\ti^{(2)}$ is a projective measurement.

To see the second claim, first we have by Bayes' Theorem that,
\begin{align*}
    &\frac{\Pr[\mathcal{M}(\regi_2, C^*, cp^*, ek_1) = 0 \big| \mathcal{EV}(\regi_1, A, ek_0, C^*, cp^*, a_1^{Can}) = 1  \wedge b^1_{Test,1} = 1\wedge b^1_{Test,2} = 1]}{\Pr[ \mathcal{M}(\regi_2, C^*, cp^*, ek_1) = 0 = 0 \big| \wedge  b^1_{Test,1} = 1\wedge b^1_{Test,2} = 1 ]} \geq \\ &\Pr[\mathcal{EV}(\regi_1, A, ek_0, C^*, cp^*, a_1^{Can}) = 1  \big|\mathcal{M}(\regi_2, C^*, cp^*, ek_1) = 0 = 0 \wedge  b^1_{Test,1} = 1\wedge b^1_{Test,2} = 1 ]
\end{align*}
The result then follows, since $\Pr[\mathcal{EV}((\regi_1, A, ek_0, C^*, cp^*), a_1^{Can}) = 1 \big| b^1_{Test,1} = 1\wedge b^1_{Test,2} = 1] \geq \frac{1}{2}\extprob$ by above and  $\Pr[ \mathcal{M}(\regi_2, C^*, cp^*, ek_1) = 0 = 0 \big| \wedge  b^1_{Test,1} = 1\wedge b^1_{Test,2} = 1 ] \geq \frac{1}{2}$  by \cref{eqn:moeimp}.

Then, applying \cref{claim:moeclaimp16} with $\mathcal{EV}_1(1^\lambda)$, we get  \begin{equation*}
        \Pr[\mathcal{E}^*_1(\regi_1, A, ek_0, C^*, cp^*) = a_1^{Can} \wedge \mathcal{E}^*_2(\regi_2, A, ek_1, C^*, cp^*) = a_2^{Can} \big| b^1_{Test,1}=1\wedge b^1_{Test,2}=1 \wedge \mathcal{M}(\regi_2) = 0] \geq (\extprob)^2
    \end{equation*}
    which completes the proof.
\end{proof}

 Since $\Pr[b^1_{Test,1}=1\wedge b^2_{Test,2}=1] \geq \frac{1}{4q_1(\lambda)}$, the claim above is a contradiction to \cref{lem:moelast} by our choice of parameters. Thus, our contradiction hypothesis that \begin{equation*}
    \Pr[b^2_{Test,1} = 0 \wedge b^2_{Test,2} = 0 \big| b^1_{Test,1} = 1 \wedge b^1_{Test,2} = 1] \geq 1 - \negl(\lambda),
\end{equation*}
where $b^1_{Test,1}, b^1_{Test,2}, b^2_{Test,1}, b^2_{Test,2} \samp \hyb_5$, is incorrect. This proves \cref{lem:moebased}.

\subsubsection{Proof of \cref{claim:moeclaimp14}}\label{sec:moeclaimp14}
We will instead prove the following stronger claim, so that the proof can serve as a proof for both\footnote{Note that while technically \cref{claim:moeclaimp16} is not a consequence of our stronger claim we prove in this section, proof of \cref{claim:moeclaimp16} follows almost exactly the same as the proof of this stronger version of \cref{claim:moeclaimp14}.} \cref{claim:moeclaimp14} and \cref{claim:moeclaimp16}.
\begin{claim}
    Let $\mathcal{M}$ be an algorithm with $1$-bit output such that
    \begin{itemize}
\item $\mathcal{M}$ runs in subexponential time $rt(\lambda)$ (precise value will be clear from the proof)
        \item $\Pr[\mathcal{M}(\regi_2, C^*, cp^*, ek_1) = 0] \geq \extprob$,
        \item $  \Pr[\ti^{(2)}(\regi_1) = 1 \big|  \mathcal{M}(\regi_2, C^*, cp^*, ek_1) = 0 \wedge b^1_{Test,1} = 1\wedge b^1_{Test,2} = 1] \geq 1 - \frac{3}{2}\cdot \sqrt{6\delta^* + 8q_1(\lambda)\cdot(\alpha_3 + \alpha_2)}\cdot\frac{2}{\extprob}$,
        \item $\Pr[\ati^{(3,1)}(\regi_1) = 0 \big|  \mathcal{M}(\regi_2, C^*, cp^*, ek_1) = 0 \wedge b^1_{Test,1} = 1\wedge b^1_{Test,2} = 1] \geq \frac{1}{4}\extprob$.
    \end{itemize} 
    
    Then, there exists an (unbounded) algorithm $\mathcal{E}^*_1$ such that
    \begin{equation*}
        \Pr[\mathcal{E}^*_1(\regi_1, A, ek_0, C^*, cp^*) = a_1^{Can} \big| b^1_{Test,1}=1\wedge b^1_{Test,2}=1 \wedge \mathcal{M}(\regi_2) = 0] \geq \extprob
    \end{equation*}
     where $\regi_1, \regi_2, A, ek_0, ek_1, C^*, cp^*, b^1_{Test,1}, b^1_{Test,2} \samp \hyb_4'$.
\end{claim}

 We now prove this claim. Let $\mathcal{M}$ be an algorithm as in the claim statement.
Define $\simati^{(4)}=\simatiwpnop{\delta^*}{\epsilon^*}{\max\{\mathcal{D}(C^*, cp^*), \mathcal{D}^{(2)}_{1}\}}{\frac{116\gamma}{128}}{\alpha_1(\lambda)}$ and let $t_1(\lambda)$ be the number of samples for $\simati^{(4)}$ as in \cref{lem:simati}. Then, by \cref{lem:simati}, we get
 \begin{itemize}
        \item $\Pr[\simati^{(4)}(s_{1,1}, \dots, s_{1,t_1(\lambda)})(\regi_1) = 1 \big|  \mathcal{M}(\regi_2, C^*, cp^*, ek_1) = 0 \wedge b^1_{Test,1} = 1\wedge b^1_{Test,2} = 1] \geq 1 - \frac{3}{2}\cdot \sqrt{6\delta^* + 8q_1(\lambda)\cdot(\alpha_3 + \alpha_2)}\cdot\frac{2}{\extprob}-\alpha_1-4\delta^* =\textcolor{blue}{\xi_1}$,
        \item $\Pr[\simati^{(4)}(s_{2,1}, \dots, s_{2,t_1(\lambda)})(\regi_1) = 0 \big|  \mathcal{M}(\regi_2, C^*, cp^*, ek_1) = 0 \wedge b^1_{Test,1} = 1\wedge b^1_{Test,2} = 1] \geq \frac{1}{4}\extprob - \alpha_1 - 4\delta^*$.
    \end{itemize} 
where $s_{1,1}, \dots, s_{1,t_1(\lambda)} \samp \sampchal(st)$ and $s_{2,1}, \dots, s_{2,t_1(\lambda)} \samp \mathcal{D}_1^{(1)}$.

Now define the following distributions for $j \in [t_1(\lambda)]$.
\begin{mdframed}
\underline{$\mathcal{D}^{(2)}_{1}$}

{\bf Hardcoded: $st, A, ek_0$}
\begin{enumerate}
    \item Sample an extractor seed as $se \samp \zo^{\ell_3(\lambda)}$.
    \item {Sample $r_1^{**} \samp \zo^{p_4(\lambda)}$.}
    \item {Set $pl_1 = 0 || A || se || r_1^{**} || j$.}
    \item Sample $x^* \samp \ace.\stegenc(ek_0, pl_1, \codeof{\sampinp(st)})$.
    \item Sample $ak, ch \samp  \sampchalfrominp(st, x^*)$.
         \item Output $ak, ch$.
\end{enumerate}
 \end{mdframed}

Now we claim the following.
\begin{claim}
    $\Pr[\simati^{(4)}(s_{2,1}, \dots, s_{2,t_1(\lambda)})(\regi_1) = 0 \big|  \mathcal{M}(\regi_2, C^*, cp^*, ek_1) = 0 \wedge b^1_{Test,1} = 1\wedge b^1_{Test,2} = 1] \geq \frac{1}{4}\extprob - \alpha_1 - 4\delta^* - ((\alpha_2 + 2\alpha_3)\cdot t_1(\lambda) - 2\alpha_3)\cdot\frac{1}{\extprob} = \textcolor{blue}{\xi_2}$
    where $s_{2,j} \samp \mathcal{D}_1^{(2,j)}$ for $j \in [t_1(\lambda)]$
\end{claim}
\begin{proof}
This claim follows easily by an $\ace$ key puncturing and iO argument, changing each sample one by one, since the program ignores the $ind$ part of the decrypted messages.
\end{proof}

By applying the hybrid lemma to above, we get that there exists some index $i^*\in\{1,\dots,t_1(\lambda)\}$ and values $\tau_1, \tau_2$ such that
\begin{itemize}
    \item $\Pr[\simati^{(4)}((sl_i)_{i \in [t(\lambda)]})(\regi_1) = 1 \big|  \mathcal{M}(\regi_2, C^*, cp^*, ek_1) = 0 \wedge b^1_{Test,1} = 1\wedge b^1_{Test,2} = 1] \geq  \tau_1$,
        \item $\Pr[\simati^{(4)}((sr_i)_{i \in [t(\lambda)]})(\regi_1) = 0 \big|  \mathcal{M}(\regi_2, C^*, cp^*, ek_1) = 0 \wedge b^1_{Test,1} = 1\wedge b^1_{Test,2} = 1] \geq  \tau_2$
        \item $\tau_2 - \tau_1 \geq \frac{ \xi_2 -  \xi_1}{t_1(\lambda)}$
\end{itemize}
where
\begin{itemize}
\item $\regi_1, \regi_2, A, ek_0, ek_1, C^*, cp^*, b^1_{Test,1}, b^1_{Test,2} \samp \hyb_4''$ and
    \item $sl_i \samp  \mathcal{D}_1^{(2,i)}$ for $i < i^*$ and $sl_i \samp \mathcal{D}(cp^*, C^*)$ otherwise,
     \item $sr_i \samp  \mathcal{D}_1^{(2,i)}$ for $i \leq i^*$ and $sr_i \samp \mathcal{D}(cp^*, C^*)$ otherwise.
\end{itemize}

Then we define the following distribution and claim the following.
\begin{mdframed}
\underline{$\mathcal{D}_{1}^{(1,j)}$}

{\bf Hardcoded: $cp^*, C^*, A, a_1^{Can}, ek, r^*_1$}
\begin{enumerate}
    \item Sample an extractor seed as $se \samp \zo^{\ell_3(\lambda)}$.
\item Set $pl_1 = 0 || A || se || (\ext(se, a_1^{Can}) \oplus r^*_1) || j$ 

\item Set $x^* \samp \ace.\stegenc(ek_0, pl_1)$.

\item Sample $ak,ch \samp \sampchalfrominp(st, x^*)$.
    
    \item Output $ak, ch$.
\end{enumerate}
\end{mdframed}
\begin{claim}
    \item $\Pr[\simati^{(4)}((sl_i)_{i \in [t(\lambda)]})(\regi_1) = 1 \big|  \mathcal{M}(\regi_2, C^*, cp^*, ek_1) = 0 \wedge b^1_{Test,1} = 1\wedge b^1_{Test,2} = 1] \geq  \tau_1  - (2\alpha_2 + \alpha_3 + 2\alpha_4)\cdot \frac{4}{\extprob}$
    where 
    \begin{itemize}
    \item $\regi_1, \regi_2, A, ek_0, ek_1, C^*, cp^*, b^1_{Test,1}, b^1_{Test,2} \samp \hyb_4'$
        \item  $sl_i \samp  \mathcal{D}_1^{(2,i)}$ for $i < i^*$
        \item $sl_i \samp \mathcal{D}_1^{(1,i)}$ if $i = i^*$,
        \item $sl_i \samp \mathcal{D}(cp^*, C^*)$ for $i > i^*$,
    \end{itemize}
\end{claim}
\begin{proof}
We will define a sequence of modifications for $\hyb_4'$.

We define the following modification $\hyb_4^{'(1)}$ of $\hyb_4^{'}$: The challenger samples $\obfd{M}$ as $\obfd{M} \samp \io(M)$ instead of $\obfd{M} \samp \io(M')$. This essentially undoes the change from $\hyb_1$ to $\hyb_2$.

By \cref{lem:shobased}, we have 
$\Pr[\simati^{(4)}((sl_i)_{i \in [t(\lambda)]})(\regi_1) = 1 \big|  \mathcal{M}(\regi_2, C^*, cp^*, ek_1) = 0 \wedge b^1_{Test,1} = 1\wedge b^1_{Test,2} = 1] \geq  \tau_1  - (\alpha_4)\cdot \frac{4}{\extprob}$
where  \begin{itemize}
    \item $\regi_1, \regi_2, A, ek_0, ek_1, C^*, cp^*, b^1_{Test,1}, b^1_{Test,2} \samp \hyb_4^{'(1)}$
        \item  $sl_i \samp  \mathcal{D}_1^{(2,i)}$ for $i < i^*$
        \item $sl_i \samp \mathcal{D}(cp^*, C^*)$ for $i \geq i^*$,
    \end{itemize}

We define $\hyb_4^{'(2)}$ so that when the challenger is preparing $\obfd{P}$, it now embeds $dk' = \ace.\mathsf{GenDK}(sk, T_{pl^*})$ instead of $dk' = \ace.\mathsf{GenDK}(sk, \outfalse)$ where \begin{itemize}
    \item $T_{pl^*}(m)$ is the circuit that outputs $\outtrue$ if $m = pl^*$,
    \item $pl^* = \ace.\enc(ek_0, 0 || A || se^* || (\ext(se^*, a_1^{Can}) \oplus r^*_1) || i^*)$ where $se^*, r_1^*$ are the values sampled during sampling $sl_{i^*} \samp \mathcal{D}_1^{(1,i^*)}$
\end{itemize}  
Since the experiments do not use an encryption of $pl^*$, we can apply the puncture-hiding security of $\ace$ to get
$\Pr[\simati^{(4)}((sl_i)_{i \in [t(\lambda)]})(\regi_1) = 1 \big|  \mathcal{M}(\regi_2, C^*, cp^*, ek_1) = 0 \wedge b^1_{Test,1} = 1\wedge b^1_{Test,2} = 1] \geq  \tau_1  - (\alpha_2+\alpha_4)\cdot \frac{4}{\extprob}$
where  \begin{itemize}
    \item $\regi_1, \regi_2, A, ek_0, ek_1, C^*, cp^*, b^1_{Test,1}, b^1_{Test,2} \samp \hyb_4^{'(2)}$
        \item  $sl_i \samp  \mathcal{D}_1^{(2,i)}$ for $i < i^*$
        \item $sl_i \samp \mathcal{D}(cp^*, C^*)$ for $i \geq i^*$,
    \end{itemize}

Now, we claim $\Pr[\simati^{(4)}((sl_i)_{i \in [t(\lambda)]})(\regi_1) = 1 \big|  \mathcal{M}(\regi_2, C^*, cp^*, ek_1) = 0 \wedge b^1_{Test,1} = 1\wedge b^1_{Test,2} = 1] \geq  \tau_1  - (2\alpha_2+\alpha_4)\cdot \frac{4}{\extprob}$
where  \begin{itemize}
    \item $\regi_1, \regi_2, A, ek_0, ek_1, C^*, cp^*, b^1_{Test,1}, b^1_{Test,2} \samp \hyb_4^{'(2)}$
        \item  $sl_i \samp  \mathcal{D}_1^{(2,i)}$ for $i < i^*$
        \item $sl_i \samp \mathcal{D}_1^{(1,i)}$ if $i = i^*$,
        \item $sl_i \samp \mathcal{D}(cp^*, C^*)$ for $i > i^*$,
    \end{itemize}
   This follows by steganographic ciphertext property of $\ace$ since $dk'$ is punctured at $pl^*$.

   We define $\hyb_4^{'(3)}$ so that when the challenger is preparing $\obfd{P}$, it now embeds $dk' = \ace.\mathsf{GenDK}(sk, \outfalse)$.
    We claim $\Pr[\simati^{(4)}((sl_i)_{i \in [t(\lambda)]})(\regi_1) = 1 \big|  \mathcal{M}(\regi_2, C^*, cp^*, ek_1) = 0 \wedge b^1_{Test,1} = 1\wedge b^1_{Test,2} = 1] \geq  \tau_1  - (\alpha_3+\alpha_4+2\alpha_2)\cdot \frac{4}{\extprob}$
where  \begin{itemize}
    \item $\regi_1, \regi_2, A, ek_0, ek_1, C^*, cp^*, b^1_{Test,1}, b^1_{Test,2} \samp \hyb_4^{'(3)}$
        \item  $sl_i \samp  \mathcal{D}_1^{(2,i)}$ for $i < i^*$
        \item $sl_i \samp \mathcal{D}_1^{(1,i)}$ if $i = i^*$,
        \item $sl_i \samp \mathcal{D}(cp^*, C^*)$ for $i > i^*$,
    \end{itemize}
    This follows by the security of $\io$ since by construction of $pl^*$, the programs have the same functionality in both cases.

    Finally, we claim  $\Pr[\simati^{(4)}((sl_i)_{i \in [t(\lambda)]})(\regi_1) = 1 \big|  \mathcal{M}(\regi_2, C^*, cp^*, ek_1) = 0 \wedge b^1_{Test,1} = 1\wedge b^1_{Test,2} = 1] \geq  \tau_1  - (\alpha_3+2\alpha_4+2\alpha_2)\cdot \frac{4}{\extprob}$
where  \begin{itemize}
    \item $\regi_1, \regi_2, A, ek_0, ek_1, C^*, cp^*, b^1_{Test,1}, b^1_{Test,2} \samp \hyb_4^{'(3)}$
        \item  $sl_i \samp  \mathcal{D}_1^{(2,i)}$ for $i < i^*$
        \item $sl_i \samp \mathcal{D}_1^{(1,i)}$ if $i = i^*$,
        \item $sl_i \samp \mathcal{D}(cp^*, C^*)$ for $i > i^*$,
    \end{itemize}
    This follows directly by \cref{lem:shobased}.
\end{proof}
The claim above shows that we can distinguish between $\mathcal{D}_1^{(1,i^*)}$ and $ \mathcal{D}_1^{(2,i^*)}$ with advantage $\geq  \frac{ \xi_2 -  \xi_1}{t_1(\lambda)} - \tau_1  - (2\alpha_2 + \alpha_3 + 2\alpha_4)\cdot \frac{4}{\extprob}$. Since this is larger than $\epsilon_{\ext}$ by our choice of parameters, we get by the reconstructive property of $\ext$ that there exists an algorithm $\mathcal{E}_1^*$ such that \begin{equation*}
        \Pr[\mathcal{E}^*_1(\regi_1, A, ek_0, C^*, cp^*) = a_1^{Can} \big| b^1_{Test,1}=1\wedge b^1_{Test,2}=1 \wedge \mathcal{M}(\regi_2) = 0] \geq \extprob
    \end{equation*}
     where $\regi_1, \regi_2, A, ek_0, ek_1, C^*, cp^*, b^1_{Test,1}, b^1_{Test,2} \samp \hyb_4'$.

     This completes the proof.

\section{Steganographic Asymmetrically Constrained Encapsulation}
In this section, we show how to construct an \emph{asymmetrically constrained encapsulation} (ACE) scheme (\cite{STOC:CHJV15}) that satisfies an additional steganographic security property: $\ace.\stegenc(ek, m, \mathcal{D})$ produces an encapsulation that is indistinguishable from a fresh sample from the distribution $\mathcal{D}$, to any adversary that has keys punctured at $m$.

\subsection{Definition}
We first recall the definition of ACE (\cite{STOC:CHJV15}).
The definition we write below is mostly verbatim from \cite{ccakan2025copy}.
\begin{definition}[Asymmetrically Constrained Encapsulation]\label{def:ace}
    Let $\mathcal{M} = \{0,1\}^\acemeslen$ denote the message space, where $\acemeslen = \poly(\lambda)$. Let $\acectcount$ and $\acecirclen$ be parameters that are polynomials in $\lambda$. We define an admissible puncturing circuit to be a circuit of size at most $\acecirclen$ and we define an admissible message to be a message in $\mathcal{M}$. Unless otherwise specified, when we consider a puncturing circuit we will always implicitly consider admissible circuits and we will always consider admissible messages.
    
    An \textbf{asymmetrically constrained encapsulation scheme} $\ace$ is parametrized by $\acecirclen, \acectcount, \acemeslen$ and consists of the following \textbf{efficient} algorithms.

\begin{itemize}
    \item $\Setup(1^\lambda)$: Takes as input the security parameter and outputs a secret key $sk$.
    \item $\genek(sk, C)$: Takes as input the secret key and a (puncturing) circuit $C$, and outputs a punctured encapsulation key $ek$.
    \item $\gendk(sk, C)$: Takes as input the secret key and a circuit $C$, and outputs a punctured decapsulation key $dk$.
    \item $\mathsf{Enc}(ek, m)$: Takes as input an encapsulation key $ek$ and a message $m$, outputs a ciphertext $ct$ or the symbol $\bot$.
    \item $\mathsf{Dec}(dk, m)$: Takes as input a decapsulation key $dk$ and a ciphertext $ct$, outputs a message $m$ or the symbol $\bot$.
\end{itemize}    
Note that all of the algorithms, except $\mathsf{Setup}$, are deterministic. Further, the key and ciphertext sizes and the runtimes of the scheme will be a polynomial $\acemeslen, \acecirclen, \acectcount$.

We require the following correctness guarantees and puncture-hiding security (defined later) and pseudorandom ciphertext security (defined later).
\paragraph{Correctness of Decapsulation:} For all (admissible) circuits $C, C'$ and messages $m$ such that $C(m) = C'(m) = \outfalse$,
\begin{equation*}
    \Pr[\mathsf{Dec}(\gendk(sk, C), \mathsf{Enc}(\genek(sk, C'), m)) = m : sk \samp \mathsf{Setup}(1^\lambda)] = 1
\end{equation*}

\paragraph{Equivalence of Constrained Encapsulation:} For all circuits $C$ and messages $m$ such that $C(m) =\outfalse$,
\begin{equation*}
    \Pr[ \mathsf{Enc}(\genek(sk, C), m) = \mathsf{Enc}(\genek(sk, \outfalse), m) : sk \samp \mathsf{Setup}(1^\lambda)] = 1
\end{equation*}
Here, the notation $\outfalse$ is overloaded to also denote the circuit that outputs $\outfalse$ everywhere.

\paragraph{Safety of Constrained Decapsulation:} For all circuits $C$ and strings $str$,
\begin{equation*}
    \Pr[ \mathsf{Dec}(dk, str) = \perp \vee~C(\mathsf{Dec}(dk, str))= \outfalse :  \begin{array}{c}
         sk \samp \mathsf{Setup}(1^\lambda)  \\
         dk = \gendk(sk, C)
    \end{array}] = 1
\end{equation*}

\paragraph{Equivalence of Constrained Decapsulation:} For all circuits $C$ and strings $str$,
\begin{equation*}
    \Pr[m_1 = m_2 \vee~ C(m_1) = \outtrue : \begin{array}{c}
         sk \samp \mathsf{Setup}(1^\lambda)  \\
         dk_2 = \gendk(sk, \outfalse) \\
         dk_1 = \gendk(sk, C)\\
         m_2 = \mathsf{Dec}(dk_2, str)\\
         m_1 = \mathsf{Dec}(dk_1, str)\\
    \end{array}] = 1
\end{equation*}
\paragraph{Unique Encapsulations:} There exists a (deterministic) mapping $\mathcal{F}$ from the secret keys and messages to ciphertexts such that 
\begin{itemize}
    \item For any $str \neq \mathcal{F}(sk, m)$, $\Pr[\mathsf{Dec}(\gendk(sk, \outfalse), str) = m] = 0$.
    \item $\Pr[\mathsf{Dec}(\gendk(sk, \outfalse), \mathcal{F}(sk, m)) = m] = 1$.
\end{itemize}
\end{definition}

Now we recall security of constrained decapsulation, also known as puncture-hiding security.

\begin{definition}[Security of Constrained Decapsulation (Puncture-Hiding Security)]
    Consider the following game between a challenger and an adversary.

    \paragraph{\underline{$\acepunchidinggame$}}
\begin{enumerate}
\item The adversary $\adve$ outputs three (admissible) circuits $C, C_0, C_1$.
\item The challenger samples $sk \samp \ace.\Setup(1^\lambda)$ and computes $ek = \ace.\genek(sk, C)$, $dk_0 = \ace.\gendk(sk, C_0)$ and $dk_1 = \ace.\gendk(sk, C_1)$.
\item The challenger samples $b \samp \zo$ and sends $ek, dk_b$ to the adversary $\adve$.
\item The adversary outputs a bit $b'$.
\item The challenger outputs  $1$ if $b=b'$, and outputs $0$ otherwise.
\end{enumerate}
    
    An asymmetrically constrained encapsulation scheme $\ace$ is said to satisfy \emph{security of constrained decapsulation} if for any QPT adversary $\adve$,
    \begin{equation*}
        \Pr[\acepunchidinggame = 1] \leq \frac{1}{2}+\negl(\lambda).
    \end{equation*}
\end{definition}

Now we recall $\acectcount$-pseudorandom ciphertext security, introduced by \cite{ccakan2025copy} to generalize $\acectcount$-ciphertext indistinguishability security requirement originally considered by \cite{STOC:CHJV15}.

\begin{definition}[Pseudorandom Ciphertext Security]
    Consider the following game between a challenger and an adversary.

    \paragraph{\underline{$\aceprct$}}
\begin{enumerate}
\item The adversary $\adve$ outputs two (admissible) circuits $C_1, C_2$ and $\acectcount$ messages $(m_1, \dots, m_{\acectcount})$ such that $C_1(m_i) = C_2(m_i) = \outtrue$ for all $i \in [\acectcount]$.
\item The challenger samples $sk \samp \ace.\Setup(1^\lambda)$ and computes $ek = \ace.\genek(sk, \outfalse)$ $ek' = \ace.\genek(sk, C_1)$, $dk' = \ace.\gendk(sk, C_2)$.
\item The challenger computes $ct_i = \ace.\mathsf{Enc}(ek, m^b_i)$ for $i \in [\acectcount]$.
\item The challenger samples strings $r_i$ for all $i \in [\acectcount]$, each the same length as  the ciphertext length of $\ace$ and each is sampled uniformly at random.
\item The challenger samples $b \samp \zo$. If $b=0$, it submits $ct_1, \dots, ct_\acectcount$ to the adversary. If $b=1$, it submits $r_1, \dots, r_\acectcount$ to the adversary.
\item The adversary outputs a bit $b'$.
\item The challenger outputs  $1$ if $b=b'$, and outputs $0$ otherwise.
\end{enumerate}
    
    An asymmetrically constrained encapsulation scheme $\ace$ is said to satisfy \emph{$\acectcount$-pseudorandom ciphertext security} if for any QPT adversary $\adve$,
    \begin{equation*}
        \Pr[\aceprct = 1] \leq \frac{1}{2}+\negl(\lambda).
    \end{equation*}
\end{definition}

Finally, we introduce our new security notion, \emph{steganographic ciphertext security}.

\begin{definition}[Steganographic Asymmetrically Constrained Encapsulation]
Let $\acectcount$, $\acecirclen$, $\acemeslen$ (as in \cref{def:ace}), $\acedistsize$, be parameters that are polynomials in $\lambda$. We define an admissible distribution $\mathcal{D}$\footnote{We will overload the notation $\mathcal{D}$ to denote both the distribution itself and the sampling circuit for this distribution.}  to be a distribution with a sampler circuit of size at most $\acedistsize$ and with min-entropy at least $\aceminent \geq 4\cdot \acemeslen$. Unless otherwise specified, we will implicitly mean an admissible distribution when we are considering a distribution for steganography.

A \textbf{steganographic asymmetrically constrained encapsulation scheme} $\ace$ is an ACE scheme parametrized by $\acedistsize$, $\acectcount$, $\acecirclen$, $\acemeslen$, $\aceminent$ with the following additional algorithms and guarantees.

\begin{itemize}
    \item $\stegenc(ek, m, \codeof{\mathcal{D}})$: Takes as input an encapsulation key $ek$, a message $m$ and a sampler circuit $\mathcal{D}$, and outputs a ciphertext.
    \item $\stegdec(dk, ct)$: An \textbf{efficient} algorithm that takes as input a decapsulation key $dk$ and a ciphertext $ct$, outputs a message $m$ or $\bot$.
\end{itemize}

We require the usual correctness requirements of an ACE scheme. Additionally, we require the following.

\paragraph{Correctness of Steganographic Encryption:} We require that for all (admissible) distributions $\mathcal{D}$ and all messages $m$, $\ace.\stegenc(ek, m, \mathcal{D})$ runs in time $\acestegbound$ (which is a function that depends on $\acemeslen, |\supp(\mathsf{D})|$ and the desired security level) and satisfies the following.
\begin{equation*}
    \Pr[\ace.\stegdec(dk, \ace.\stegenc(ek, m, \codeof{\mathcal{D}})) = m : \begin{array}{c}
           sk \samp \ace.\Setup(1^\lambda)\\
           ek = \ace.\genek(sk, \outfalse) \\
         dk = \ace.\gendk(sk, \outfalse)
    \end{array}] \geq 1 - \negl(\lambda).
\end{equation*}

\paragraph{Security:} We require that $\ace$ satisfies puncture-hiding security and $\acectcount$-pseudorandom ciphertext security against any quantum adversary that runs in time $\poly(\lambda)\cdot\acestegbound$.

\paragraph{Steganographic Ciphertext Security:} We require that $\ace$ satisfies steganographic ciphertext security (defined later).
\end{definition}

\begin{definition}[Steganographic Ciphertext Security]
  Consider the following game between a challenger and an adversary.

    \paragraph{\underline{$\acesteg$}}
\begin{enumerate}
\item The adversary $\adve$ outputs two admissible circuits $C_1, C_2$ and $\acectcount$ messages $(m_1, \dots, m_{\acectcount})$ such that $C_1(m_i) = C_2(m_i) = \outtrue$ for all $i \in [\acectcount]$.
\item The adversary outputs an (admissible) sampler $\mathcal{D}$.
\item The challenger samples $sk \samp \ace.\Setup(1^\lambda)$ and computes $ek = \ace.\genek(sk, \outfalse)$ $ek' = \ace.\genek(sk, C_1)$, $dk' = \ace.\gendk(sk, C_2)$.
\item The challenger computes $ct_i = \ace.\stegenc(ek, m_i, \codeof{\mathcal{D}})$ for $i \in [\acectcount]$.
\item The challenger samples strings $samp_i \samp \mathcal{D}$ for all $i \in [\acectcount]$.
\item The challenger samples $b \samp \zo$. If $b=0$, it submits $ct_1, \dots, ct_\acectcount$ to the adversary. If $b=1$, it submits $samp_1, \dots, samp_\acectcount$ to the adversary.
\item The adversary outputs a bit $b'$.
\item The challenger outputs  $1$ if $b=b'$, and outputs $0$ otherwise.
\end{enumerate}
    
     An asymmetrically constrained encapsulation scheme $\ace$ is said to satisfy \emph{steganographic ciphertext security} if for any QPT adversary $\adve$,
    \begin{equation*}
        \Pr[\acesteg = 1] \leq \frac{1}{2}+\negl(\lambda).
    \end{equation*}
\end{definition}

\begin{theorem}
Assuming the existence of subexponentially secure indistinguishability obfuscation and one-way functions, for any polynomials $\acemeslen, \acecirclen, \acedistsize, \acectcount$, there exits a subexponentially secure steganographic asymmetrically constrainable encryption scheme.
\end{theorem}
\begin{proof}
    We give our construction in \cref{sec:acecons}. The proof of steganographic ciphertext security is given in \cref{sec:aceproof}. Steganographic correctness is immediate. All the other correctness and security notions follow from \cite{ccakan2025copy}. All the assumed primitives can be instantiated assuming existence of subexponentially secure indistinguishability obfuscation and one-way functions for any subexponential function, as explained in \cref{sec:acecons}, \cref{sec:aceproof} and \cite{ccakan2025copy}.
\end{proof}

\subsection{Construction}\label{sec:acecons}
In this section, we give our construction for a steganographic ACE scheme. Our construction is similar to the construction \cite{STOC:CHJV15}, except for two key differences: (i) addition of the steganographic encryption/decryption algorithms (which naturally also requires a novel security proof), (ii) choosing the underlying primitives and their parameters in a delicate way that is required for our new security notions.

\paragraph{\underline{$\ace.\Setup(1^\lambda)$}}
\begin{enumerate}
\item Sample $K_1 \samp \prf_1.\keygen(1^\lambda)$ and $K_2 \samp \prf_2.\keygen(1^\lambda)$.
\item Sample $seed \samp \zo^{\aceselen}$.
\item Output $sk = (K_1, K_2, seed)$.
\end{enumerate}

\paragraph{\underline{$\ace.\genek(sk, C)$}}
\begin{enumerate}
\item Parse $(K_1, K_2, seed) = sk$.
\item Sample $\obfd{P}\samp \io(P)$ where $P$ is the following program.
\begin{mdframed}
    {\bf \underline{$P(m)$}}
    
     {\bf Hardcoded: $K_1, K_2, C$}
    \begin{enumerate}[label=\arabic*.]
        \item Check if $C(m) = \outtrue$. If so, output $\bot$ and terminate.
        \item Compute $\alpha = \prf_1.\ceval(K_1, m)$.
        \item Compute $\beta = \prf_2.\ceval(K_2, \alpha) \oplus m$.
        \item Output $\alpha || \beta$.
    \end{enumerate}
    \end{mdframed}
    \item Output $(\obfd{P}, seed)$.
\end{enumerate}

\paragraph{\underline{$\ace.\gendk(sk, C)$}}
\begin{enumerate}
\item Parse $(K_1, K_2, seed) = sk$.
\item Sample $\obfd{P}\samp \io(P)$ where $P$ is the following program.
\begin{mdframed}
    {\bf \underline{$P(ct)$}}
    
     {\bf Hardcoded: $K_1, K_2, C$}
    \begin{enumerate}[label=\arabic*.]
        \item Check if $C(m) = \outtrue$. If so, output $\bot$ and terminate.
        \item Parse $\alpha || \beta = ct$.
        \item Compute $m = \prf_2.\ceval(K_2, \alpha)\oplus \beta$.
        \item Check if $\alpha = \prf_1.\ceval(K_1, m)$. If not, output $\bot$ and terminate.
        \item Output $m$.
    \end{enumerate}
    \end{mdframed}
    \item Output $(\obfd{P}, seed)$.
\end{enumerate}

\paragraph{\underline{$\ace.\enc(ek, m)$}}
\begin{enumerate}
\item Parse $(\obfd{P}, seed) = ek$.
\item Output $\obfd{P}(m)$.
\end{enumerate}

\paragraph{\underline{$\ace.\dec(dk, ct)$}}
\begin{enumerate}
\item Parse $(\obfd{P}, seed) = dk$.
\item Output $\obfd{P}(ct)$.
\end{enumerate}

\paragraph{\underline{$\ace.\stegenc(ek, m)$}}
\begin{enumerate}
\item Parse $(\obfd{P}, seed) = ek$.
\item Set $ict = \ace.\enc(ek, m)$.
\item Set $cnt = 1$.
\item Repeat the following as long as $cnt \leq \ell$: Sample $s' \samp \mathcal{D}$. If $\ext(seed, s') = ict$, output $s'$ and terminate, otherwise increase $cnt$ by one and continue.
\item Output $\bot$ if not already terminated.
\end{enumerate}

\paragraph{\underline{$\ace.\stegdec(dk, ct)$}}
\begin{enumerate}
\item Parse $(\obfd{P}, seed) = dk$.
\item Output $\obfd{P}(\ext(seed, ct)).$
\end{enumerate}

\subsection{Proof of Steganographic Ciphertext Security}\label{sec:aceproof}
We first start with two technical lemmata.
\begin{lemma}[Infinite Reverse Resampling Lemma]\label{lem:infrev}
    Let $\mathcal{D}$ be any distribution, $S$ be any set and $f:\supp(\mathcal{D})\to S$ be any function. Consider the distribution $\mathcal{D}'$ defined as follows.
    \begin{mdframed}
    {\bf \underline{$\mathcal{D}'$}}
    \begin{enumerate}[label=\arabic*.]
        \item Sample $s \samp \mathcal{D}$.
        \item Set $y = f(s)$.
        \item Repeat the following until success: Sample $s' \samp \mathcal{D}$. If $f(s') = y$, output $s'$ and terminate, otherwise continue.
    \end{enumerate}
    \end{mdframed}
    Then, $\mathcal{D} \equiv \mathcal{D}'$.
\end{lemma}
\begin{proof}
    Fix any $x \in \supp(\mathcal{D})$ and we will show that $\Pr_{\mathcal{D}'}[x] = \Pr_{\mathcal{D}}[x]$, which implies $\mathcal{D} \equiv \mathcal{D}'$.

    We write \begin{align*}
        \Pr_{\mathcal{D}'}[x] &= \sum_{s \in \supp(\mathcal{D})} \Pr_{\mathcal{D}}[s]\cdot\lim_{k\to\infty}\left(\sum^k_{i = 0}(\Pr_{s_1\samp \mathcal{D}}[f(s_1) \neq f(s)])^{i}\Pr_{s_2\samp D}[s_2 = x \wedge f(x) = f(s)]\right) \\ &=\sum_{s \in \supp(\mathcal{D})} \Pr_{\mathcal{D}}[s]\cdot\Pr_{s_2\samp D}[s_2 = x \wedge f(x) = f(s)]\cdot\frac{1}{1-\Pr_{s_1\samp \mathcal{D}}[f(s_1) \neq f(s)]}
         \\ &=\sum_{s \in \supp(\mathcal{D})} \Pr_{\mathcal{D}}[s]\cdot\Pr_{s_2\samp D}[s_2 = x \big| f(s_2) = f(s)]
         \\ &=\sum_{y \in f(\supp(\mathcal{D}))}\sum_{\substack{s \in \supp(\mathcal{D}):\\ f(s) = y}} \Pr_{\mathcal{D}}[s]\cdot\Pr_{s_2\samp D}[s_2 = x \big| f(s_2) = f(s)]
         \\ &=\sum_{y \in f(\supp(\mathcal{D}))}\sum_{\substack{s \in \supp(\mathcal{D}):\\ f(s) = y}} \Pr_{\mathcal{D}}[s]\cdot\Pr_{s_2\samp D}[s_2 = x \big| f(s_2) = y]
         \\ &=\sum_{y \in f(\supp(\mathcal{D}))} \Pr_{s_2\samp D}[s_2 = x \big| f(s_2) = y] \cdot\left(\sum_{\substack{s \in \supp(\mathcal{D}):\\ f(s) = y}} \Pr_{\mathcal{D}}[s]\right)
         \\ &=\sum_{y \in f(\supp(\mathcal{D}))} \Pr_{s_2\samp D}[s_2 = x \big| f(s_2) = y] \cdot \Pr_{s \samp\mathcal{D} }[f(s) = y] = \Pr_{s_2 \samp \mathcal{D}}[s_2 = x]
    \end{align*}
    The first line is by the Law of Total probability and the second line is by the power series solution and the other lines should be self-evident.
\end{proof}

\begin{lemma}[Truncated Reverse Resampling Lemma]\label{lem:trunc}
 Let $\mathcal{D}$ be any distribution, $S$ be any set and $f:\supp(\mathcal{D})\to S$ be any function. Let $\eps > 0$ and $\acelimit = \lceil\frac{2\left(\log(4)n + \log(\frac{1}{\epsilon})\right)\cdot |\supp(\mathcal{D})|)}{\eps}\rceil$. Consider the distribution $\mathcal{D}'$ defined as follows.
    \begin{mdframed}
    {\bf \underline{$\mathcal{D}'$}}
    \begin{enumerate}[label=\arabic*.]
        \item Sample $s \samp \mathcal{D}$.
        \item Set $y = f(s)$.
        \item Set $cnt = 1$.
        \item Repeat the following as long as $cnt \leq \acelimit$: Sample $s' \samp \mathcal{D}$. If $f(s') = y$, output $s'$ and terminate, otherwise increase $cnt$ by one and continue.
        \item Output $\bot$ if not already terminated.
    \end{enumerate}
    \end{mdframed}
    Then, $\mathcal{D} \approx_\eps \mathcal{D}'$ and $\Pr_{\mathcal{D}'}[f(s) = y] \geq 1 - \eps$.
\end{lemma}
\begin{proof}
    Define $E$ to be the event that during execution of the sampler $\mathcal{D}'$, the preimage search succeeds, that is, the counter does not reach $\acelimit + 1$. We write \begin{align*}
       |\mathcal{D} - \mathcal{D}'|_1 &= \sum_{x \in \supp(\mathcal{D})} |\Pr_{D}[x] - \Pr_{D'}[x]| \\ &=
       \sum_{x \in \supp(\mathcal{D})} |\Pr_{D}[x] - \Pr_{x' \samp D'}[x' = x | E]\cdot\Pr_{\mathcal{D}'}[E] - \Pr_{x' \samp D'}[x' = x | \overline{E}]\cdot\Pr_{\mathcal{D}'}[\overline{E}]|
       \\ &=
       \sum_{x \in \supp(\mathcal{D})} |\Pr_{D}[x] - \Pr_{x' \samp D'}[x = x' | E]\cdot\Pr_{\mathcal{D}'}[E]|
       \\ &=
       \sum_{x \in \supp(\mathcal{D})} |\Pr_{D}[x] - \Pr_{x' \samp D}[x = x']\cdot\Pr_{\mathcal{D}'}[E]|
       \\ &=
       (1-\Pr_{\mathcal{D}'}[E])\cdot\sum_{x \in \supp(\mathcal{D})} \Pr_{D}[x]
    \end{align*}
    The third line is due to the fact that when $\overline{E}$ occurs, output is $\bot$. The fourth line is by \cref{lem:infrev}.

    Set $th = \frac{\eps}{2\cdot |\supp(\mathcal{D})|}$. Now we have \begin{align*}
        \Pr_{\mathcal{D}'}[\overline{E}] &\leq \sum_{x \in \supp(\mathcal{D})} \Pr_{\mathcal{D}}[x]\cdot (1-\Pr_{\mathcal{D}}[x])^{\acelimit} \\ &= \sum_{\substack{x \in \supp(\mathcal{D}):\\ \Pr_{\mathcal{D}}[x]  \leq th}} \Pr_{\mathcal{D}}[x]\cdot (1-\Pr_{\mathcal{D}}[x])^{\acelimit} + \sum_{\substack{x \in \supp(\mathcal{D}):\\ \Pr_{\mathcal{D}}[x]  > th}} \Pr_{\mathcal{D}}[x]\cdot (1-\Pr_{\mathcal{D}}[x])^{\acelimit}  \\ &\leq |\supp(\mathcal{D})|\cdot (th + e^{-\acelimit\cdot th}) \leq \eps
    \end{align*}
\end{proof}

Now we prove security through a sequence of hybrids, each of which is obtained by modifying the previous one. For simplicity we will only prove the case $\ell = 1$, but the general case follows similarly (given that obfuscated program sizes will depend on $\ell$). We first define our hybrids and then we will show their indistinguishability.

\paragraph{\underline{$\hyb_{0}$:}} The original game $\acesteg$.

\paragraph{\underline{$\hyb_{1}$:}} We modify the way the challenger computes the punctured key $ek'$ that will be submitted to the adversary: Instead of setting $ek = \ace.\genek(sk, C_1)$, the challenger now sets $ek' = (\obfd{P}_{enc}, seed)$ where \begin{itemize}
\item $K'_1 \samp \prf_1.\punc(K_1, m^*)$.
\item $\alpha^* = \prf_1.\ceval(K_1, m^*)$.
\item $\beta^* = \prf_1.\ceval(K_2, \alpha^*)$. (Note that $\ace.\enc(ek,m)$ is $\alpha^* || \beta^*$ during execution of $\stegenc$)
     \item $\obfd{P} \samp \io(P'_{enc})$ where $P'_{enc}$ is the following program.
    \begin{mdframed}
    {\bf \underline{$P'_{enc}(m)$}}
    
     {\bf Hardcoded: $\textcolor{red}{K_1'}, K_2, C_1$}
    \begin{enumerate}[label=\arabic*.]
        \item Check if $C_1(m) = \outtrue$. If so, output $\bot$ and terminate.
        \item Compute $\alpha = \prf_1.\ceval(\textcolor{red}{K_1'}, m)$.
        \item Compute $\beta = \prf_2.\ceval(K_2, \alpha) \oplus m$.
        \item Output $\alpha || \beta$.
    \end{enumerate}
    \end{mdframed}
\end{itemize}

\paragraph{\underline{$\hyb_{2}$:}} We modify the way the challenger computes the punctured key $dk'$ that will be submitted to the adversary: Instead of setting $dk = \ace.\gendk(sk, C_2)$, the challenger now sets $dk' = (\obfd{P}_{dec}, seed)$ where  $\obfd{P} \samp \io(P'_{dec})$ where $P'_{dec}$ is the following program.
   \begin{mdframed}
    {\bf \underline{$P'_{dec}(ct)$}}
    
     {\bf Hardcoded: $\textcolor{red}{K_1'}, K_2, C$}
    \begin{enumerate}[label=\arabic*.]
        \item Check if $C(m) = \outtrue$. If so, output $\bot$ and terminate.
        \item Parse $\alpha || \beta = ct$.
        \item Compute $m = \prf_2.\ceval(K_2, \alpha)\oplus \beta$.
        \item Check if $\alpha = \prf_1.\ceval(\textcolor{red}{K_1'}, m)$. If not, output $\bot$ and terminate.
        \item Output $m$.
    \end{enumerate}
    \end{mdframed}

\paragraph{\underline{$\hyb_{3}$:}} We now sample $\alpha^* \samp \zo^{3\acemeslen} \setminus \mathsf{Img}(\prf_1.\ceval(K_1, \cdot))$.

\paragraph{\underline{$\hyb_{4}$:}} When sampling $ek'$, we now sample $K_2' \samp \prf_2.\punc(K_2, \{\alpha^*, r_1^*\})$ where $r_1^* \samp \zo^{3\acemeslen}\setminus \mathsf{Img}(\prf_1.\ceval(K_1, \cdot))$ and use $K_2'$ instead of $K_2$.

\paragraph{\underline{$\hyb_{5}$:}}  We modify the way the challenger computes the punctured key $dk'$ that will be submitted to the adversary: Instead of setting $dk = \ace.\gendk(sk, C_2)$, the challenger now sets $dk' = (\obfd{P}_{dec}, seed)$ where  $\obfd{P} \samp \io(P''_{dec})$ where $P''_{dec}$ is the following program.
   \begin{mdframed}
    {\bf \underline{$P''_{dec}(ct)$}}
    
     {\bf Hardcoded: ${K_1'}, K_2, C,\textcolor{red}{\alpha^*, r_1^*}$}
    \begin{enumerate}[label=\arabic*.]
        \item Check if $C(m) = \outtrue$. If so, output $\bot$ and terminate.
        \item \textcolor{red}{Check if $\alpha = \alpha^*$ or $r_1^*$. If so, output $\bot$ and terminate.}
        \item Parse $\alpha || \beta = ct$.
        \item Compute $m = \prf_2.\ceval(K_2, \alpha)\oplus \beta$.
        \item Check if $\alpha = \prf_1.\ceval({K_1'}, m)$. If not, output $\bot$ and terminate.
        \item Output $m$.
    \end{enumerate}
    \end{mdframed}

    \paragraph{\underline{$\hyb_{6}$:}} When sampling $dk'$, we now sample $K_2' \samp \prf_2.\punc(K_2, \{\alpha^*, r_1^*\})$ where $r_1^* \samp \zo^{3\acemeslen}$ and use $K_2'$ instead of $K_2$.

    \paragraph{\underline{$\hyb_{7}$:}}  We now sample $\beta^* \samp \zo^{\acemeslen}$.

    \paragraph{\underline{$\hyb_{8}$:}}  We now sample $r_1^* \samp \zo^{3\acemeslen}$.
    
\begin{lemma}
    $\hyb_0\approx \hyb_1$.
\end{lemma}
\begin{proof}
    The circuits in the two hybrids are equivalent by punctured key correctness of the PRF scheme. The result follows by $\io$ security.
\end{proof}

\begin{lemma}
    $\hyb_1\approx \hyb_2$.
\end{lemma}
\begin{proof}
    The circuits in the two hybrids are equivalent by punctured key correctness of the PRF scheme. The result follows by $\io$ security.
\end{proof}

\begin{lemma}
    $\hyb_2\approx_{2^{-2\acemeslen}} \hyb_3$.
\end{lemma}
\begin{proof}
 Follows by punctured key security of the PRF scheme  and the fact that the image set of $\prf_1.\ceval_1(K_1,\cdot)$ has size $2^{\acemeslen}$.
\end{proof}

\begin{lemma}
    $\hyb_3\approx \hyb_4$.
\end{lemma}
\begin{proof}
    Follows by PRF punctured key correctness and $\io$ security.
\end{proof}

\begin{lemma}
    $\hyb_4\approx \hyb_5$.
\end{lemma}
\begin{proof}
    Follows by $\io$ security.
\end{proof}

\begin{lemma}
    $\hyb_5\approx \hyb_6$.
\end{lemma}
\begin{proof}
    Follows by PRF punctured key correctness and $\io$ security.
\end{proof}

\begin{lemma}
    $\hyb_6\approx \hyb_7$.
\end{lemma}
\begin{proof}
    Follows by PRF punctured key security.
\end{proof}
\begin{lemma}
    $\hyb_7\approx_{2^{-2\acemeslen}} \hyb_8$.
\end{lemma}
\begin{proof}
    Immediate by an elementary probability calculation.
\end{proof}

The above argument shows that in the experiment $\acesteg$, during the computation of $\ace.\stegenc$, we can replace the value $ict = \ace.\enc(ek, m)$ with a truly random value $r$. Then, we do the following argument. We replace $r$ with $\ext(seed, s^*)$ where $s^* \samp \mathcal{D}$, and these hybrids are indistinguishable by extractor security. Finally, we can replace the outcome of $\ace.\stegenc$ with a true sample from $s''$ from $\mathcal{D}$ by \cref{lem:trunc}.

\section{Acknowledgments}
Part of this work done was done while A\c{C} was an intern at NTT Research and part of this work was done while A\c{C} was being supported by the following grants of VG: NSF award 1916939, DARPA SIEVE program, a gift from Ripple, a DoE NETL award, a JP Morgan Faculty Fellowship, a PNC center for financial services innovation award, and a Cylab seed funding award.

\appendix
\bibliographystyle{alpha}
\bibliography{abbrev0,crypto,additional}

\end{document}